\documentclass[12pt]{iopart}

\usepackage{graphicx, amssymb,amsthm, epsfig, amsbsy, amsfonts}

\newtheorem{theorem}{Theorem}
\newtheorem{lemma}{Lemma}

\theoremstyle{definition}

\renewcommand\L{\mathbf{L}}
\def\N{\mathbb{N}}

\newcommand{\ui}{\textrm{i}}
\newcommand{\ue}{\textrm{e}}
\newcommand{\UI}{\mathbf{I}}

\newcommand{\ud}{\mathrm{d}}

\newcommand{\gz}{{\mathbb Z}}

\newcommand{\re}{\Re}

\newcommand{\vp}{\boldsymbol{\psi}}
\newcommand{\vc}{\mathbf{c}}

\newcommand{\vz}{\mathbf{0}}

\newcommand{\rk}{\mathrm{rank}}

\newcommand{\A}{{\mathbb A}}
\newcommand{\B}{{\mathbb B}}

\newcommand{\SM}{{\mathbb S}}

\newcommand{\bdm}{\begin{displaymath}}
\newcommand{\edm}{\end{displaymath}}
\newcommand{\beq}{\begin{equation}}
\newcommand{\eeq}{\end{equation}}
\newcommand{\beqa}{\begin{eqnarray}}
\newcommand{\eeqa}{\end{eqnarray}}
\newcommand{\nn}{\nonumber}

\newcommand{\cL}{\mathcal{L}}

\newcommand{\csch}{\mathrm{csch}}
\newcommand{\sech}{\mathrm{sech}}

\begin{document}

\title{Zeta functions of quantum graphs}

\author{JM Harrison and K Kirsten}

\address{Department of Mathematics, Baylor University, Waco, TX 76798, USA}

\ead{jon\_harrison@baylor.edu; klaus\_kirsten@baylor.edu}
\begin{abstract}

In this article we construct zeta functions of quantum graphs using a contour integral technique based on the argument principle.  We start by considering the special case of the star graph with Neumann matching conditions at the center of the star.  We then extend the technique to allow any matching conditions at the center for which the Laplace operator is self-adjoint and finally obtain an expression for the zeta function of any graph with general vertex matching conditions.  In the process it is convenient to work with new forms for the secular equation of a quantum graph that extend the well known secular equation of the Neumann star graph. In the second half of the article we apply the zeta function to obtain new results for the spectral determinant, vacuum energy and heat kernel coefficients of quantum graphs.  These have all been topics of current research in their own right and in each case this unified approach significantly expands results in the literature.
\end{abstract}

\maketitle

\section{Introduction}

Zeta functions associated with combinatorial graphs have been well studied.
The Ihara zeta function \cite{p:I:DS2*2PLG} was interpreted in terms of a finite graph by Sunada \cite{p:S:LFG}, for subsequent generalizations see e.g. \cite{p:H:AFFGRPG, p:B:ISZFTL, p:ST:ZFFGC}.
  Such combinatorial graph zeta functions are defined by an Euler product over sets of primitive cycles (periodic orbits) which are itineraries of graph edges.
%
A quantum graph identifies bonds (edges) of a combinatorial graph with closed intervals 
generating a metric graph.   In addition  a quantum graph has an operator acting on functions defined on the collection of intervals.  This is typically a self-adjoint Hamiltonian operator acting on functions in a Hilbert space defined on the metric graph; in the context of inverse spectral problems see, e.g., \cite{carl99-351-4069}.  Quantum graphs  were introduced to model electrons in organic molecules, some more recent applications include areas of mesoscopic physics, quantum chaos, photonic crystals, superconductivity, the quantum Hall effect, microelectronics and the theory of waveguides, see \cite{p:K:GMWPTS,p:K:QG:IBS} for recent surveys of the applications of graph models.  In such applications the spectrum of the quantum graph determines physical properties of the model.  In this sense the quantum graph is not far removed from its combinatorial cousin where important graph theoretic properties like the connectivity of the graph are approached through the spectrum of a discrete Laplace operator on the graph.

Spectral questions asked of quantum graphs are traditionally approached through a trace formula.  Trace formulae express spectral functions like the density of states or heat kernel as sums over periodic orbits on the graph.  They naturally acquire a similar flavor to the Ihara zeta function.  The first graph trace formula was derived by Roth \cite{p:R:SLSG}. Kottos and Smilansky \cite{p:KS:QCG, p:KS:POTSSQG} introduced a contour integral approach to the trace formula starting with a secular equation based on the scattering matrix of plane-waves on the graph.  Solutions of the secular equation correspond to points in the spectrum of the quantum graph.  (For the current state of the art of graph trace formulae see \cite{p:BE:TFQGGSABC}.)

If the spectrum of the quantum graph is $0\leqslant \lambda_0 \leqslant \lambda_1 \leqslant \dots $ then formally the associated spectral zeta function $\zeta(s)$ is defined by
\begin{equation}\label{eq:spec zeta}
    \zeta(s)={\sum_{j=0}^{\infty}}\phantom{|}^{\prime} \lambda_j^{-s} \ ,
\end{equation}
where the prime indicates that zero modes, if present,  are omitted from the summation.
Using the trace formula such a zeta function can immediately be written as a sum over periodic orbits.  However, adopting this approach the zeta function and the trace formula are closely related and the zeta function does not provide new insight.  Instead we take a step back from the trace formula and derive the zeta function using a contour integral based on new secular equations for the graph.  Applying a contour transformation we obtain an
integral formulation for the zeta function based on the matching conditions at the graphs vertices directly in the general case.
The program can be thought of as a generalization of the special case of a star graph with Neumann matching conditions at the
vertices (defined in Section \ref{sec:zeta fns neumann star}).  In this case there is a well known secular equation whose
solutions define the spectrum of the star \cite{p:BK:TPSCSG, p:BBK:SGSB},
\begin{equation}\label{eq:secular example}
    \sum_{b=1}^{B} \tan k L_b =0 \ .
\end{equation}
The sum is over the $B$ arms of the star and $L_b$ is the length of the $b$'th arm.  If $k_j$ is a solution of (\ref{eq:secular example}) then $\lambda_j=k_j^2$ is an eigenvalue of the Laplace operator on the star.  In such cases the simple structure of the secular equation produces more explicit formulations of the zeta function.

The zeta function of a particular infinite quantum graph whose spectrum is related to Dirichlet's divisor problem is studied in \cite{p:ES:SIQG} and the zeta function of the Berry-Keating operator appears in \cite{p:ES:BKOQG}, however there was previously no general formulation of zeta functions of the Laplace operator on graphs.
In addition many spectral properties of quantum graphs, that are subjects of active investigation \cite{p:ACDMT:SDQG, p:BHW:MAVEQG,p:BE:TFQGGSABC, p:D:SDSOG, p:D:SDGGBC, p:F:DSOMG, p:FKW:VERCFQSG, p:KMW:VDESDQSG,p:KPS:HKMGTF,p:R:SLSG,p:T:ZRSDMG}, can be obtained in a straightforward way from the zeta function.  In the second half of the article we apply our representations of the zeta function to derive the spectral determinant, vacuum energy and heat kernel asymptotics of quantum graphs.  In all these cases the zeta function approach adopted here yields new results.  In detail we obtain the following theorems.
\begin{theorem}\label{thm:gen spec det}
For the Laplace operator on a graph whose vertex matching conditions are defined by a pair of matrices $\A$ and $\B$, with $\A\B^\dagger=\B\A^\dagger$ and $\rk (\A,\B)=B$, the spectral determinant is
\bdm
    {\det}' (-\triangle)= \frac{2^B \hat{f}(0)}{c_N \prod_{b=1}^B L_b} \ ,
\edm
where
\bdm
    \hat{f}(t)= \det \left( \A
    - t \B
    \left( \begin{array}{cc}
    \coth(t \L) & -\csch (t \L)\\
    -\csch (t \L)&  \coth(t \L)\\
    \end{array} \right)
    \right) \ .
\edm
$c_N$ is the coefficient of the leading order term $t^{-N}$ in the asymptotic expansion of $\hat{f}(t)$ as $t\to \infty$.
\end{theorem}

\begin{theorem}\label{thm:Fc graph}
For the Laplace operator on a graph, under the conditions of Theorem \ref{thm:gen spec det}, the Casimir force on the bond $\beta$ is
\begin{equation*}
    F_c^\beta =\frac{\pi}{24 L_\beta^2} + \frac{1}{\pi}\int_0^\infty \frac{\partial}{\partial L_\beta} \log \hat{f}(t) \, \ud t
\end{equation*}
provided the graph is generic: the poles of $f(z)=\hat{f}(-\ui z)$ are the whole of the set $\{ m\pi/L_b | m\in \gz, b=1,\dots,B \}$.
\end{theorem}
In this article we follow the scheme introduced by Kostrykin and Schrader \cite{p:KS:KRQW} to classify matching conditions of self-adjoint Laplace operators on the graph, so matching conditions at the vertices are defined by a matrix equation $\A \vp + \B \vp' =\vz$ where $\vp$ and $\vp'$ are vectors of the values of a function $\psi$ and its outgoing derivatives at the ends of the bonds, see Section \ref{sec:quantum graphs}.
Theorems \ref{thm:gen spec det} and \ref{thm:Fc graph} provide new general formulations of the spectral determinant and vacuum energy for the Laplace operator on a graph with
any vertex matching conditions illustrating the power of this unified approach.

The article is structured as follows.  In Section \ref{sec:quantum graphs}
we define the quantum graph models that we study.  Section
\ref{sec:zeta fns neumann star} introduces the contour integral approach used to evaluate the graph zeta functions with the simplest example of a star graph with Neumann matching conditions at the center.  Section \ref{sec:general star} generalizes the star graph results by allowing any form of matching conditions at the center consistent with a self-adjoint realization of the Laplace operator on the star.  Having demonstrated the techniques used to incorporate general matching conditions at a vertex we formulate the zeta function of general quantum graphs in terms of the vertex matching conditions in Section \ref{sec:general zeta} which is the central result of the paper.  As a byproduct we also formulate a number of new forms for the secular equation of a quantum graph.

In the second half of the paper we investigate the implications of our results for a number of spectral quantities.  Section \ref{sec:spectral det} compares the results for the spectral determinant derived from the zeta function to those already obtained for quantum graphs.  The zeta function approach developed here is not only more direct but the results obtained have a particularly simple form being expressed directly in terms of matching conditions.  Section
\ref{sec:vacuum energy} uses the zeta functions to obtain a new formulation of graph vacuum energy.  Section \ref{sec:heat kernel} determines implications of the zeta function results for the asymptotics of the heat kernel. In Section \ref{sec:piston} we provide a concrete and current application of our results to a graph theoretic generalization of a piston.
In the Conclusions we point out the most important results of our contribution.
\section{Quantum graph model}\label{sec:quantum graphs}
 For the purpose of this article the particular quantum graph model we consider is a self-adjoint Laplace operator on a metric graph.  However, the term quantum graph is often applied more widely to describe self-adjoint differential or pseudo-differential operators on metric graphs.
%
In this section we introduce the graph models we employ, for a general review of analysis on quantum graphs see \cite{p:GS:QG:AQCUSS, p:K:QG:I}.

A graph is a set of vertices connected by bonds, see for example Figure \ref{fig:star}.  In a metric graph $G$ each bond $b$ is associated with an interval $[0,L_b]$ so $L_b$ is the length of $b$.
For a bond $b=(v,w)$ connecting vertices $v$ and $w$ the choice of orientation for the coordinate $x_b$ on the interval $[0,L_b]$ is arbitrary, our results are independent of this choice of orientation.  However, for the sake of clarity, when a bond is written as a pair of vertices $b=(v,w)$ the coordinate $x_b=0$ at $v$ and $x_b=L_b$ at $w$.
The vertices and bonds are enumerated so $v\in \{1,2,\dots ,V\}$ and $b\in \{1,2,\dots, B\}$.   The total length of the graph $G$ is denoted by $\cL=\sum_{b=1}^B L_b$.

A quantum graph consists of a metric graph with a self-adjoint differential operator on the set of intervals associated with the graph bonds.  In this article we consider Laplace operators on metric graphs.
The differential operator on the bonds of the graph is
$-\frac{\ud^2}{\ud x_b^2}$.
%
%
%
%
A function $\psi$ on $G$ is defined by the set of functions $\{ \psi_b(x_b) \}_{b=1,\dots,B}$ on the intervals associated to the bonds. The Hilbert space for the graph is consequently
\begin{equation}\label{eq:Hilbert space}
    {\mathcal H} = \bigoplus_{b=1}^B L^2 \bigl( [0,L_b] \bigr) \ .
\end{equation}
A self-adjoint realization of the Laplace operator on $G$ is determined by specifying a suitable domain in ${\mathcal H}$.  This can be achieved by defining an appropriate set of matching conditions at the vertices of $G$.  We assume all matching conditions are local: at a vertex $v$ matching conditions respect the connectivity of the graph only relating values of the function and its derivatives at the ends of the intervals connected at $v$.
General matching conditions at all vertices of $G$ are specified by a pair of $2B\times 2B$ matrices $\A$ and $\B$.  For a function $\psi$ on $G$ let
\begin{eqnarray}\label{eq:vertex values}
    \vp&=&\big(\psi_1(0),\dots,\psi_B(0),\psi_1(L_1),\dots,\psi_B(L_B)\big)^T \ , \\
    \vp'&=&\big(\psi'_1(0),\dots,\psi'_B(0),-\psi'_1(L_1),\dots,-\psi'_B(L_B)\big)^T \ .
\end{eqnarray}
The graphs matching conditions are then defined by the matrix equation
\begin{equation}\label{eq:matching conditions}
    \A \vp + \B \vp'=\vz \ .
\end{equation}
The following theorem of Kostrykin and Schrader classifies all matching conditions of self-adjoint realizations of the Laplace operator \cite{p:KS:KRQW}.  (An alternative unique classification scheme was introduced by Kuchment in \cite{p:K:QG:I}.)
\begin{theorem}\label{thm:sa matching conditions}
The Laplace operator with matching conditions specified by $\A$ and $\B$ is self-adjoint if and only if $(\A,\B)$ has maximal rank and $\A\B^\dagger=\B\A^\dagger$.
\end{theorem}
Although we set out a general classification scheme for self-adjoint Laplace operators on graphs it will not be used until Section  \ref{sec:general star}, initially we concentrate on the simpler case of a star graph with Neumann like matching conditions at the center of the star.

Given a self-adjoint Laplace operator $-\triangle$ on a graph we are interested in properties of the spectrum $\lambda_0 \leqslant \lambda_1 \leqslant \lambda_2 \leqslant \dots$.  It is convenient to introduce an alternative spectral parameter $k$ so we study solutions of the eigenproblem
\begin{equation}\label{eq:eigen problem}
    -\triangle \psi = k^2 \psi \ .
\end{equation}
Then $\lambda_j=k_j^2$ and we refer to the non-negative sequence $0\leqslant k_0\leqslant k_1 \leqslant k_2 \leqslant \dots$ as the $k$-spectrum.  Formally the spectral zeta function of a quantum graph that will be our object of study is,
\begin{equation}\label{eq:k-spec zeta}
    \zeta(s)={\sum_{j=0}^{\infty}}\phantom{|}^{\prime} \, k_j^{-2s} \ .
\end{equation}

\section{Zeta functions of star graphs with Neumann matching conditions at the center}
\label{sec:zeta fns neumann star}
To demonstrate the technique we start with a model problem of a star graph with Neumann like matching conditions at the center.
This type of model has been studied in a number of settings, for example analyzing spectral statistics \cite{p:BBK:SGSB, p:BK:TPSCSG, p:KS:QCG, p:KS:POTSSQG} and the distribution of wavefunctions \cite{p:BKW:IWFS, p:BKW:NQESG}.  In general, the results obtained for this particular model are more explicit than equivalent results for other graphs so it
provides an appropriate jumping off point for our investigation.

A star graph consists of a central vertex of degree $B$ which we refer to as the \emph{center} and $B$ external vertices of degree one which we call the \emph{nodes} of the star, see Figure \ref{fig:star}.
The boundary conditions at the nodes will be either Neumann or Dirichlet respectively, where we take the edges to be oriented so $x_b=0$ at the node. $\psi'_n(0)=0$ or $\psi_d(0)=0$ where $\{n\}, \{d\}$ label the sets of bonds starting at the $B_N$ Neumann or $B_D$ Dirichlet nodes respectively.  The total number of bonds of the star is then $B=B_N+B_D$ and we will still use $\{b\}$ to label the set of all bonds.
At the center we choose matching conditions that generalize the Neumann boundary conditions: functions are continuous at the center, $\psi_b(L_b)=\phi$ for all $b$, and
\begin{equation}\label{eq:Neumann mc}
    \sum_b \psi'_b(L_b)=0 \ .
\end{equation}
$\phi$ is not a constant but only convenient notation for the single value of the wave function at the central vertex.


\begin{figure}[htb]
\begin{center}
\includegraphics[width=4cm]{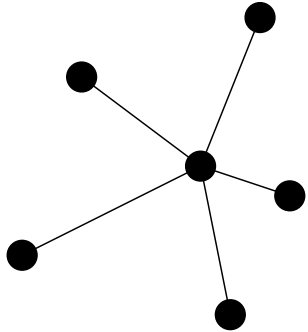}
\caption{A star (or hydra) graph.}\label{fig:star}
\end{center}
\end{figure}

\subsection{The secular equation.}
The Laplace operator on a star graph with Neumann matching conditions at the central vertex admits a secular equation, of a
particularly simple form, whose solutions form the $k$-spectrum of the star \cite{p:GS:QG:AQCUSS}.
From the boundary condition at the nodes and the continuity of the wave function at the central vertex,
eigenfunctions on the bonds of the star have the form
\begin{equation}\label{eq:star wave fn}
    \psi_n(x_n)=\phi \frac{\cos k x_n }{\cos k L_n} \ , \qquad
    \psi_d(x_d)=\phi \frac{\sin k x_d }{\sin k L_d} \ ,
\end{equation}
where $k^2$ is the eigenvalue of the Laplace operator.  Substituting in (\ref{eq:Neumann mc}) produces a secular equation,
\begin{equation}\label{eq:secular star}
    \sum_n \tan k L_n - \sum_d \cot k L_d = 0 \ ,
\end{equation}
where the sums are over the bonds starting at Neumann/Dirichlet nodes respectively.
The positive set of solutions $\{k_j \}$ is the $k$-spectrum of the graph.

If we consider the function
\begin{equation}\label{eq:defn f1}
    f(k)= \sum_n \tan k L_n - \sum_d \cot k L_d \ ,
\end{equation}
$f$ has poles at the points
\begin{equation}\label{eq:n poles}
    \bigcup_n \left\{  \frac{\pi(m+1/2)}{L_n} \right \}_{m\in \mathbb{Z}} \cup \bigcup_d \left\{ \frac{\pi m}{L_d} \right\}_{m\in \mathbb{Z}} \ .
\end{equation}
 If the set of bond lengths $\{ L_1, \dots , L_B \}$ is incommensurate, not rationally related, the poles are all distinct.  As both $\tan$ and negative $\cot$ are strictly increasing the zeros of $f$, which correspond to eigenvalues of the Laplace operator, are also all distinct each one lying between a pair of adjacent poles.  Figure
 \ref{fig:secular fn} shows a schematic representation of $f$.

\begin{figure}[htb]
  \begin{center}
    \setlength{\unitlength}{1cm}
    \begin{picture}(8,6)
    \put(0,0){\includegraphics[width=8cm]{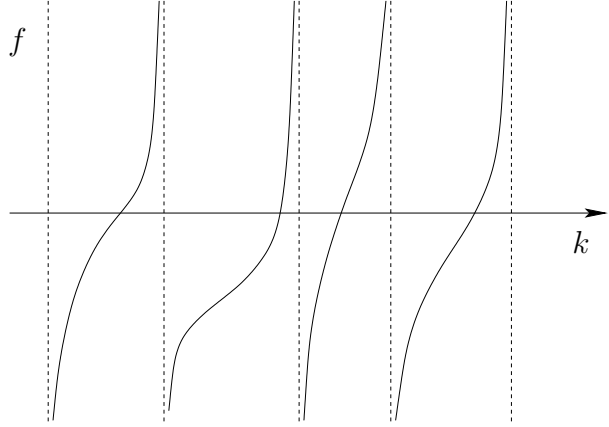}}
    \put(7.5,2.3){$k$}
    \put(0,5){$f$}
    \end{picture}
    \caption{Schematic representation of the functional part of the secular equation (\ref{eq:defn f1}) of the star graph with Neumann matching at the center.}\label{fig:secular fn}
  \end{center}
\end{figure}

\subsection{The zeta function for Neumann boundary conditions at the nodes.}
As a first step we derive the zeta function of the star where all the nodes have Neumann boundary conditions.  Let
\begin{equation}\label{eq:defn f n}
    f(z)=\frac{1}{z} \sum_n \tan z L_n \ ,
\end{equation}
where $z=k+\ui t \in \mathbb{C}$.
The zeros of $f$ on the positive real axis still correspond to square roots of the eigenvalues of the graph Laplacian, while dividing by $z$ removes the zero of $f$ at the origin.
To formulate our graph zeta function (\ref{eq:spec zeta}) we follow a contour integral approach introduced in \cite{p:KM:FDCIM, p:KM:FDGSLP, p:KL:CDCI}.  Each zero $k_j$ of $f$ contributes a factor $k_j^{-2s}$ to the zeta function and to sum these we use the argument principle \cite{b:C:FOCVI} and evaluate an integral of the form
\begin{equation}\label{eq:contour int}
    \zeta(s)=\frac{1}{2\pi \ui} \int_c z^{-2s} \frac{f'(z)}{f(z)}\,  \ud z = \frac{1}{2\pi \ui} \int_c z^{-2s} \frac \ud {\ud z} \log f(z)\,  \ud z \ ,
\end{equation}
where $c$ encloses the zeros of $f$ and avoids poles which we have already seen are distinct when the bond lengths are incommensurate.
%
Figure \ref{fig:star contour}(a) illustrates the appropriate form of the contour $c$. \footnote[1]{Of course, the contour $c$ has to be thought of as being the limit of a finite contour $c_n$ as $n\to\infty$. The contour $c_n$ is as in Figure 3(a) but closed by a vertical line
with real part $a_n$. The sequence $(a_n)_{n\in\N}$, which goes to infinity, is chosen such that the distance of each $a_n$ to the nearest pole of $\log f(z)$
is larger than a suitably chosen $\epsilon >0$. This guarantees that $\log f(z)$ remains bounded along the vertical lines and so the contributions
from the vertical lines vanish for $\Re s > 1/2$ as $n\to\infty$. The sequence $a_n$ can be constructed because of
Weyls law, as the average separation of the $k$-spectrum is constant.} 
%
%

\begin{figure}[htb]
  \begin{center}
    \setlength{\unitlength}{1cm}
    \begin{picture}(12,4)
    \put(1,0){\includegraphics[width=4cm]{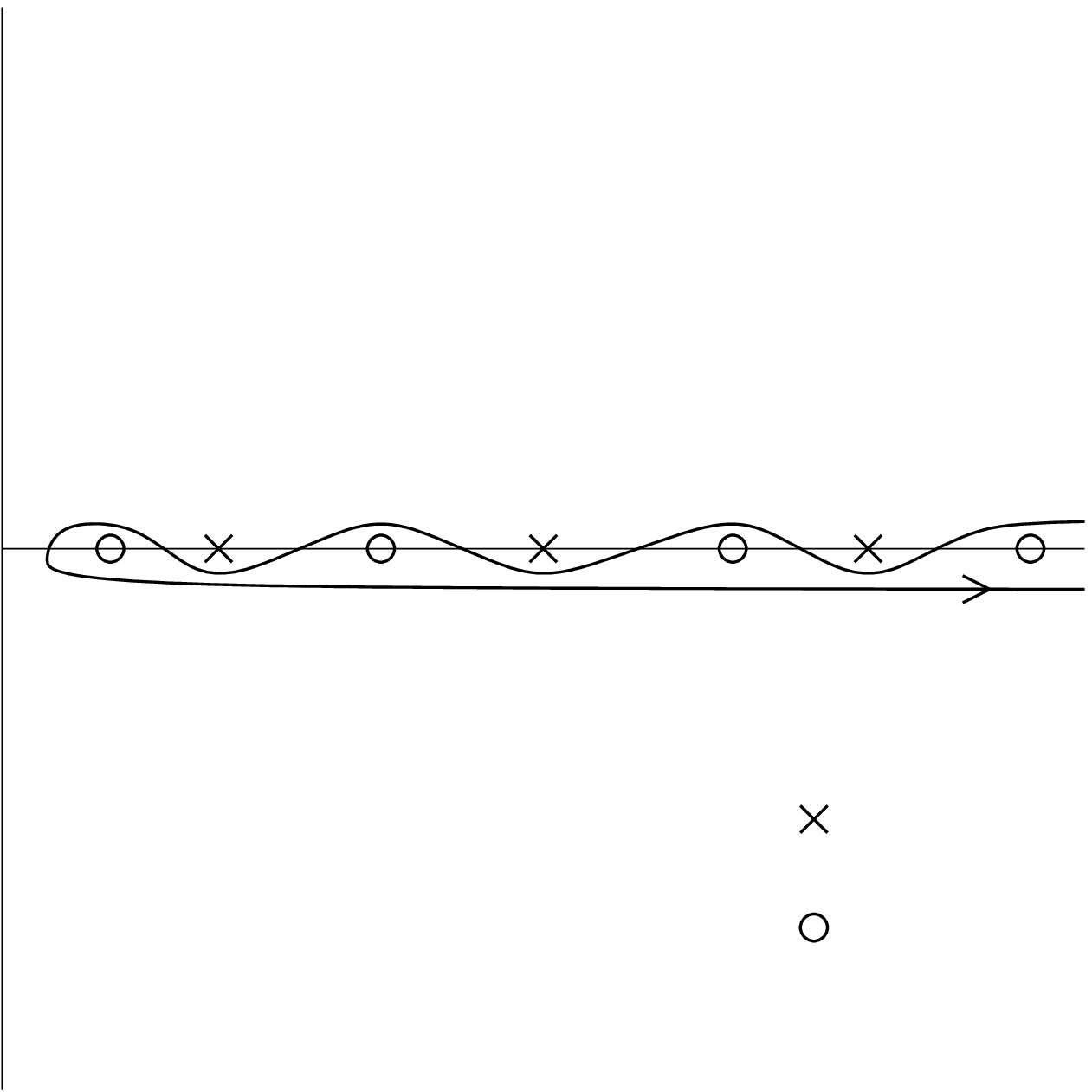}}
    \put(8,0){\includegraphics[width=4cm]{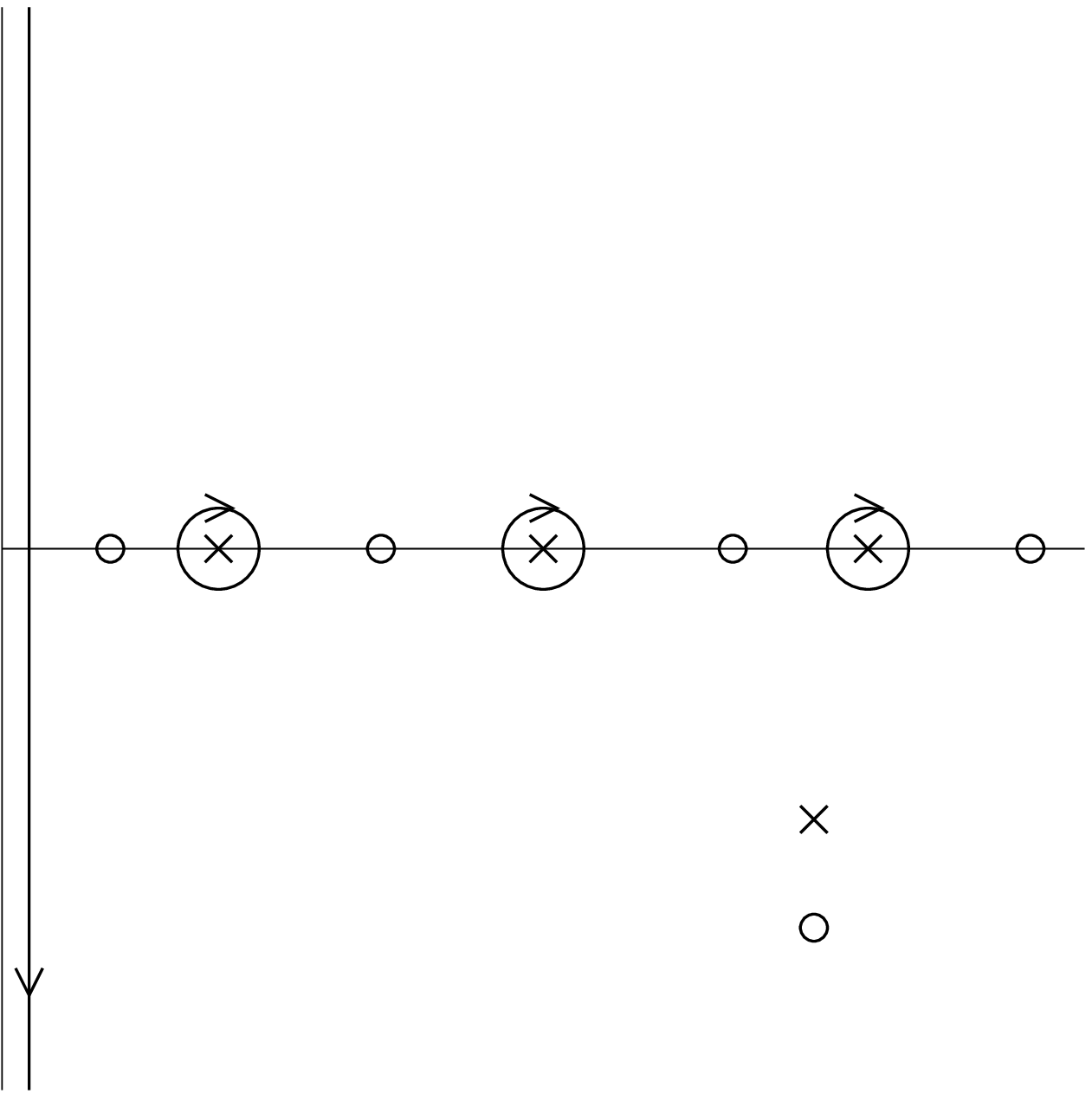}}
    \put(0,3.5){(a)}
    \put(7,3.5){(b)}
    \put(4.2,1.5){$c$}
    \put(8.3,0.5){$c'$}
    \put(4.2,0.9){pole}
    \put(4.2,0.5){zero}
    \put(11.2,0.9){pole}
    \put(11.2,0.5){zero}
    \end{picture}
    \caption{The contours used to evaluate the star graph zeta function, (a) before, and (b) after, the contour transformation.}\label{fig:star contour}
  \end{center}
\end{figure}

To analyze $\zeta(s)$ we deform $c$ to $c'$ and integrate along the imaginary axis, see Figure \ref{fig:star contour} (b).
\footnote{The contour $c'$ is closed by a semicircular arc of radius $a_n$ on the right. Along the arc, $f(z)$ grows at most exponentially fast and as $n\to\infty$ the relevant ratio $|f'(z)/f(z)|$ is asymptotically constant and contributions from the contour at infinity vanish for $\Re s >1/2$. Although we will not stress this again, the same construction holds for all zeta functions considered later.}

%
%
%


Following the contour transformation it is natural to write $\zeta(s)$ as the sum of two terms
\begin{equation}\label{eq:Im plus poles}
    \zeta(s)=\zeta_{Im}(s) +\zeta_{P}(s) \ ,
\end{equation}
where $\zeta_{Im}$ is the contribution generated by the integral along the imaginary axis and $\zeta_{P}$
is the series over residues arising from the poles of $f$.
At a pole $z_0$ of $f$ we must subtract the residue $z_0^{-2s}$.  Consequently
\begin{eqnarray}\label{eq:zeta P}
    \zeta_{P}(s) & = &\sum_n \left( \frac{\pi}{L_n} \right)^{-2s} \sum_{m=0}^{\infty}
    (m+1/2)^{-2s} \nn \\
    & = &\zeta_{H} (2s,1/2) \sum_n \left( \frac{\pi}{L_n} \right)^{-2s} \ ,
\end{eqnarray}
where $\zeta_H$ is the Hurwitz zeta function.

For the integral along the imaginary axis we get
\begin{eqnarray}
    \zeta_{Im}(s)&=& \frac{1}{2\pi \ui} \int_\infty^{-\infty} (\ui t)^{-2s} \frac{\ud}{\ud t} \log f(\ui t) \,  \ud t \nn \\
    &=& \frac{\sin \pi s}{\pi} \int_0^{\infty} t^{-2s} \frac{\ud}{\ud t} \log \left( \frac{\hat{f}(t)}{t} \right) \,  \ud t \ , \label{eq:zeta Im}
\end{eqnarray}
where
\begin{equation}\label{eq:defn f hat n}
    \hat{f}(t)=\sum_n \tanh t L_n \ .
\end{equation}
As $t\to 0$, $\hat{f}(t)/t \sim \sum_n (L_n -L_n^3 t^2/3) + O(t^4)$ and in this limit
$\frac{\ud}{\ud t} \log \hat{f}(t)/t$ is proportional to $t$ so (\ref{eq:zeta Im}) is valid at most for
$\re \, s<1$. Similarly, from the $t\to\infty$ behavior one obtains the restriction
$\Re \, s >0$.

Our representation of $\zeta_{Im}(s)$ therefore holds in the strip $0<\re \, s <1$.
To obtain an analytic continuation valid for $\re \, s < 1$ we split the integral at $t=1$ and develop the integral over $[1,\infty)$,
\begin{equation}\label{eq:zeta Im2}
\fl    \zeta_{Im}(s)=\frac{\sin \pi s}{\pi} \left[ \int_0^{1} t^{-2s} \frac{\ud}{\ud t} \log \left( \frac{\hat{f}(t)}{t} \right) \,  \ud t
    +\int_1^{\infty} t^{-2s} \frac{\ud}{\ud t} \log \left( \hat{f}(t) \right) \,  \ud t - \frac{1}{2s}
    \right].
\end{equation}
Collecting these results we obtain the following theorem
\begin{theorem}\label{thm:zeta nstar}
For the Laplace operator on a star graph with Neumann matching conditions at the vertices the zeta function for $\re \, s <1$ is given by
\begin{eqnarray}
\fl \zeta(s)=\zeta_{H} (2s,1/2) \sum_n \left( \frac{\pi}{L_n} \right)^{-2s}
-\frac{\sin \pi s}{2\pi s}
+\frac{\sin \pi s}{\pi} \int_0^{1} t^{-2s} \frac{\ud}{\ud t} \log \left( \frac{\hat{f}(t)}{t} \right) \,  \ud t \nn \\
    + \frac{\sin \pi s}{\pi} \int_1^{\infty} t^{-2s} \frac{\ud}{\ud t} \log \left( \hat{f}(t) \right) \,  \ud t \ . \nn
\end{eqnarray}
\end{theorem}

To demonstrate the power of this formulation of the zeta function we calculate the derivative of $\zeta$ at zero.  We discuss the implication of the result for the spectral determinant of quantum graphs in Section \ref{sec:spectral det}.
Differentiating (\ref{eq:zeta Im2}) we see that
\begin{eqnarray}
   \zeta_{Im}'(0)&=& \left[ \log \left( \frac{\hat{f}(t)}{t} \right) \right]_0^1 + \left[ \log \hat{f}(t) \right]_1^\infty \nn \\
   &=& -\log \left( \frac{\cL}{B} \right) \label{eq:zeta' Im} \ ,
\end{eqnarray}
where
$\cL/B$ is the mean bond length.  From (\ref{eq:zeta P}) the pole contribution to $\zeta'(0)$ is given by
\begin{eqnarray}\label{eq:zeta' P}
    \zeta_{P}'(0) & = & 2B \zeta_{H}' (0,1/2) + \zeta_{H} (0,1/2)
    \sum_n \left[-2\log \left( \frac{\pi}{L_n} \right)\right] \nn \\
    &=& -B\log 2 \ ,
\end{eqnarray}
where we have used $\zeta_{H} (0,1/2)=0$ and
$2\zeta_{H}' (0,1/2)=-\log(2)$.
Combining the results for $\zeta'_Im (0)$ and $\zeta'_P(0)$, we find
\begin{equation}\label{eq:zeta' 0}
    \zeta'(0)=-\log \left( \frac{2^B \cL}{B} \right) \ .
\end{equation}

\subsection{Mixed Dirichlet and Neumann conditions at the nodes.}
\label{sec:mixed nodes}
If we include $B_D$ nodes with Dirichlet boundary conditions the zeta function construction can be modified to generate an integral representation of $\zeta(s)$.
Let us define the functions
\begin{eqnarray}\label{eq:defn f dn*}
    f(z)&=&z\left( \sum_n \tan z L_n - \sum_d \cot z L_d \right)\ ,\\ \label{eq:defn fhat dn*}
    \hat{f}(t)&=& \sum_n \tanh t L_n + \sum_d \coth t L_d  \ ,
\end{eqnarray}
so that $f(\ui t)=-t\hat{f}(t)$.  $f$ is defined so that zeros of the secular equation
are zeros of $f$ but $f$ is not divergent at zero, $f(0)= -\sum_d L_d^{-1}$.
We represent the zeta function again using the same contour integral
\begin{equation}\label{eq:contour int 2}
    \zeta(s)=\int_c z^{-2s} \frac{f'(z)}{f(z)} \,  \ud z = \int_c z^{-2s} \frac \ud {\ud z} \log f(z) \,  \ud z .
\end{equation}
%
Proceeding as described before, the integral along the imaginary axis is given by

\begin{equation}
    \zeta_{Im}(s)=
    \frac{\sin \pi s}{\pi} \int_0^{\infty} t^{-2s} \frac{\ud}{\ud t} \log \left( t \hat{f}(t) \right) \,  \ud t \ ,
\end{equation}
for $0< \re \, s < 1$.  Splitting the integral at $t=1$, we obtain the analytical continuation
\begin{equation}\label{eq:dn * zeta Im}
\fl    \zeta_{Im}(s)=\frac{\sin \pi s}{\pi} \left[ \int_0^{1} t^{-2s} \frac{\ud}{\ud t} \log \Big( t\hat{f}(t) \Big) \,  \ud t
    +\int_1^{\infty} t^{-2s} \frac{\ud}{\ud t} \log  \hat{f}(t) \,  \ud t + \frac{1}{2s}
    \right] \ ,
\end{equation}
valid for $\re \, s<1$.

The positive poles of $\cot zL_d$ are $\{ m\pi/L_d \}_{m\in \mathbb{N}}$.
Summing their contributions as well as those corresponding to poles of $\tan zL_n$ discussed previously
\begin{eqnarray}\label{eq:dn* zeta P}
    \zeta_{P}(s) &=& \sum_n \Big( \frac{\pi}{L_n} \Big)^{-2s} \sum_{m=0}^{\infty}
    (m+1/2)^{-2s} + \sum_d \Big( \frac{\pi}{L_d} \Big)^{-2s} \sum_{j=1}^{\infty}
    j^{-2s} \ , \nn \\
    &=& (2^{2s}-1)\zeta_{R} (2s) \sum_n \Big( \frac{\pi}{L_n} \Big)^{-2s} + \zeta_{R}(2s) \sum_d \Big( \frac{\pi}{L_d} \Big)^{-2s} \ .
\end{eqnarray}
Collecting these results we have the following theorem.
\begin{theorem}\label{thm:zeta mixed star}
For the Laplace operator on a star graph with Neumann matching conditions at the central vertex, $B_D$ external nodes with Dirichlet boundary conditions and $B_N$
external nodes with Neumann boundary conditions the zeta function for $\Re s <1$ is given by
\begin{eqnarray}
\fl \zeta(s)=(2^{2s}-1)\zeta_{R} (2s) \sum_n \Big( \frac{\pi}{L_n} \Big)^{-2s} + \zeta_{R}(2s) \sum_d \Big( \frac{\pi}{L_d} \Big)^{-2s} + \frac{\sin \pi s}{2\pi s} \nn \\
 + \frac{\sin \pi s}{\pi} \int_0^{1} t^{-2s} \frac{\ud}{\ud t} \log \Big( t\hat{f}(t) \Big) \,  \ud t + \frac{\sin \pi s}{\pi} \int_1^{\infty} t^{-2s} \frac{\ud}{\ud t} \log  \hat{f}(t) \,  \ud t \ . \nn
\end{eqnarray}
\end{theorem}

Although the addition of bonds with Dirichlet boundary conditions appears a small variation of the star graph model -- certainly the formulation of the zeta function in theorems \ref{thm:zeta nstar} and \ref{thm:zeta mixed star} are very similar -- their addition has a substantial impact on spectral properties of the star.
This can be seen if we evaluate $\zeta'(0)$. First we note
\begin{eqnarray}
    \zeta_{Im}'(0)&=& \int_0^{1} \frac{\ud}{\ud t} \log \left( t\hat{f}(t) \right) \,  \ud t
    +\int_1^{\infty} \frac{\ud}{\ud t} \log  \hat{f}(t) \,  \ud t \ ,\nn \\
    &=& -\log \Big( \sum_d L_d^{-1} \Big) +\log B  \label{eq:dn* Im zeta' 0}.
\end{eqnarray}
Differentiating $\zeta_{P}$,
\begin{eqnarray}
    \zeta'_{P}(0)&=&2\log 2 \, \zeta_{R}(0) B_N + 2 \zeta_{R}'(0) B_D
    -2\zeta_{R}(0) \sum_d \log \Big(  \frac{\pi}{L_d} \Big) \nn \\
    &=& -\log \left[ 2^B \Big( \prod_d L_d \Big) \right] \ . \label{eq:dn* P zeta' 0}
\end{eqnarray}
Combining the results
\begin{equation}\label{eq:dn* zeta' 0}
    \zeta'(0)=-\log \left[ \frac{2^B}{B} \Big( \prod_d L_d \Big) \Big( \sum_d L_d^{-1} \Big) \right] \ .
\end{equation}
In comparison with (\ref{eq:zeta' 0}) the spectral determinant of a star with Dirichlet and Neumann nodes only depends on the lengths of the bonds starting at nodes with Dirichlet boundary conditions.
%

\subsection{Zeta functions with equal bond lengths.}
It will be instructive, for comparison, to evaluate the zeta functions of the star when all the bonds have an equal length.  In this case the spectrum is known and we can calculate the zeta function directly.  The results in this case will be more explicit and will provide a useful test case when we apply the zeta function to derive the vacuum energy of the graph.

If we consider the star graph where all the nodes have Neumann boundary conditions
and when the bond lengths are equal, $L_b=L$ for all $b$, the secular equation (\ref{eq:secular star}) reduces to $B\tan kL=0$.  The $k$-spectrum now consists both of the zeros of $\tan kL$, which are simple eigenvalues, and the poles of
$\tan kL$, which are eigenvalues with multiplicity $B-1$.  To see that the poles are also eigenvalues one may start with the secular equation $\sum_b \tan kL_b=0$ where the bond lengths are incommensurate. Initially every eigenvalue lies between a pair of adjacent poles.  As the lengths of the bonds are equalized groups of $B$ poles come together trapping $B-1$ zeros of the secular equation at a pole of $\tan kL$.  From this spectrum the zeta function excluding the zero mode is
\begin{eqnarray}
    \zeta(s)&=& \sum_{n=1}^{\infty} \left( \frac{n\pi}{L} \right)^{-2s} + (B-1) \sum_{m=0}^{\infty} \left( \frac{(2m+1)\pi}{2L} \right)^{-2s} \ , \nn \\
    &=&\left( \frac{\pi}{L} \right)^{-2s} \big( (B-1)2^{2s} -B+2 \big) \, \zeta_R(2s) \ .
    \label{eq:n* L zeta}
\end{eqnarray}
%
Consequently
\begin{eqnarray}
    \zeta'(0)&=&-2\log \left(\frac{\pi}{L}\right) \, \zeta_R(0) + (B-1)2\zeta_R(0)\log 2
    +2\zeta_R'(0) \nn \\
    &=& \log \left(\frac{\pi}{L}\right) - (B-1)\log 2 - \log 2\pi \nn \\
    &=& -\log \left( 2^B L \right) \label{eq:n* L zeta'} .
\end{eqnarray}
This agrees with our previous result with incommensurate bond lengths  (\ref{eq:zeta' 0}) as the total graph length is $\cL=BL$.  This is expected as the previous result is continuous with respect to the bond lengths.

%

If we include $B_D$ nodes with Dirichlet boundary conditions keeping all the bond lengths equal we can still evaluate the spectrum directly.  The secular equation (\ref{eq:secular star}) reduces to,
\begin{equation}
    B_N \sin^2 kL - B_D \cos^2 kL = 0 \ .
\end{equation}
If we define $\alpha =\frac{1}{\pi} \arcsin \sqrt{B_D/B}$ then zeros of the equation are values of $k$ in the set $\{ (m\pi \pm \alpha \pi)/L , m\in \gz\}$ each element of which is a simple eigenvalue. As in the previous case making the bond lengths equal traps eigenvalues at the poles of $\tan kL$ and $\cot kL$.  The sets
$\{ (2m+1)\pi/2L,  m\in \gz\}$ and $\{ m\pi/L, m\in \gz\}$ therefore correspond to eigenvalues of multiplicity $B_N-1$ and $B_D-1$ respectively.

The zeta function can now be calculated directly from the spectrum. We have
\begin{eqnarray}\label{eq:dn* zeta L}
\fl    \zeta(s)&=& \sum_{m_1=0}^{\infty} \Big( \frac{(m_1+\alpha) \pi}{L} \Big)^{-2s}
    + \sum_{m_2=1}^{\infty} \Big( \frac{(m_2-\alpha) \pi}{L} \Big)^{-2s} \nn \\
\fl    &&\qquad+(B_N-1)\sum_{m_3=0}^{\infty} \Big( \frac{(2m_3+1) \pi}{2L} \Big)^{-2s}
    +(B_D-1)\sum_{m_4=1}^{\infty} \Big( \frac{m_4\pi}{L} \Big)^{-2s}   \nn \\
\fl    &=& \Big( \frac{\pi}{L} \Big)^{-2s} \left[ \zeta_{H}(2s,\alpha) +
    \zeta_{H}(2s,1-\alpha)
    + \Big( B_D-B_N + (B_N-1)2^{2s}\Big) \zeta_{R}(2s) \right] \ .
\end{eqnarray}
From this representation of the zeta function it is straightforward to derive $\zeta'(0)$,
\begin{eqnarray}
  \zeta'(0) &=& (B_D-1) \log \Big( \frac{\pi}{L} \Big) - (B_N-1) \log 2 -(B_D-1) \log(2\pi)  \nn \\
  && + 2 \log \Big( \csc(\alpha \pi) \Big) -2\log 2 \ ,
\end{eqnarray}
where we have used
\begin{eqnarray}
  \zeta_{H}(0,\alpha) +\zeta_{H}(0,1-\alpha)  &=& 0 \\
  \zeta_{H}'(0,\alpha) +\zeta_{H}'(0,1-\alpha) &=&
  \log \big(\csc(\alpha \pi)\big)-\log 2 \ .
\end{eqnarray}
From the definition of $\alpha$ we see that $\csc (\alpha \pi)=\sqrt{B/B_D}$, consequently
\begin{eqnarray}
  \zeta'(0) &=& -B \log 2 - (B_D-1) \log(L) +\log B - \log B_D \nn \\
  &=& - \log \left( \frac{2^B L^{B_D} B_D}{\cL} \right) \ ,
\end{eqnarray}
where $\cL=BL$ is the total length of the graph.  This agrees with (\ref{eq:dn* zeta' 0}) for incommensurate bond lengths if we set $L_b=L$ for all bonds $b$.

\section{The zeta function for a star graph with general matching conditions at the center}
\label{sec:general star}

 We introduce the techniques used to study the zeta function of any quantum graphs by generalizing the star graph example to a star with Dirichlet conditions at the nodes but where the matching conditions at the center have any general form compatible with the domain of a self-adjoint Laplace operator.  The matching conditions at the central vertex will be specified by a pair of $B\times B$ matrices $\A$ and $\B$ using the scheme of Kostrykin and Schrader \cite{p:KS:KRQW}, see Section
 \ref{sec:quantum graphs}.  The Laplace operator on the star is then self-adjoint if and only if $(\A,\B)$ has maximal rank and $\A\B^\dagger=\B\A^\dagger$.

\subsection{The secular equation.}
Dirichlet boundary conditions at the nodes imply that an eigenfunction has the form
\begin{equation}
    \psi_b(x_b)=c_b \sin k x_b
\end{equation}
on each bond $b$.
The matching condition at the center is defined by the matrix equation
$\A \vp + \B \vp' = \vz $,
where
    $\vp=(\psi_1(L_1),\dots, \psi_B(L_B))^T$ and
    $\vp'=(-\psi_1'(L_1),\dots, -\psi_B'(L_B))^T$.
Let $\vc=(c_1,\dots,c_B)^T$ and define two diagonal $B\times B$ matrices
\begin{eqnarray}\label{eq:gen star c & s}
    \sin (k\L)&=&\textrm{diag}\{\sin k L_1,\dots, \sin k L_B \}  \ , \\
    \cos (k\L)&=&\textrm{diag}\{\cos k L_1,\dots, \cos k L_B \} \ .
\end{eqnarray}
The matching condition can equivalently be written as
\begin{equation}\label{eq:gen star secular a}
    \Big( \A \sin(k\L) - k \B \cos(k\L) \Big)\vc = \vz \ .
\end{equation}
$k$ is therefore an eigenvalue if and only if it is a solution of the secular equation
\begin{equation}\label{secular star gen center}
    \det \Big( \A \sin(k\L) - k \B \cos (k\L) \Big)=0 \ .
\end{equation}

We can put this in an alternate form that mirrors the form used for the zeta function of a general Laplace operator on a cone \cite{p:KLP:FDGSAELOGC,p:KLP:EEPPZFCM}.
\begin{lemma}\label{lem:useful}
Let $X=\mathrm{diag}\{x_1,\dots,x_B\}$ and $Y=\mathrm{diag}\{ y_1,\dots,y_B \}$ with both $X$ and $Y$ invertible.  Then
\begin{equation*}
    \det \left( \begin{array}{cc}
    \A & \B \\
    X & Y \\ \end{array} \right) = \det \Big( \A X^{-1} - \B Y^{-1} \Big) \prod_{j=1}^B x_j y_j \ .
\end{equation*}
\end{lemma}
This is a straightforward consequence of
\begin{equation}\label{eq:proof of lemma 1}
  \det \left( \left( \begin{array}{cc}
    \A & \B \\
    X & Y \\
    \end{array} \right) \left( \begin{array}{cc}
    X^{-1} & 0 \\
    0 & Y^{-1} \\
    \end{array} \right) \right) =
    \det \left( \begin{array}{cc}
    \A X^{-1} & \B Y^{-1} \\
    \UI_B & \UI_B \\
    \end{array} \right) 
    \ .
\end{equation}

Applying Lemma \ref{lem:useful} to the secular equation (\ref{secular star gen center}) when neither $\sin kL_b$
or $\cos kL_b$ is zero for any $b$ we obtain
\begin{equation}\label{eq:secular star gen center b}
    \det \left( \begin{array}{cc}
    \A & \B \\
    (\sin(k\L))^{-1} & \frac{1}{k}\, (\cos(k\L))^{-1} \\
    \end{array} \right) =0 \ .
\end{equation}
Equivalently we may write the secular equation in a form
that bears a functional similarity to the cases we have already analyzed when the center had Neumann like matching conditions, namely
\begin{equation}\label{eq:secular star gen center c}
    \det \left( \begin{array}{cc}
    \A & \B \\
    \UI_B & \frac{1}{k} \, \tan (k\L) \\
    \end{array} \right) =0 \ .
\end{equation}

\subsection{Zeta function calculation.}
Following the final formulation of the secular equation let us define functions
\begin{equation}\label{eq:star gc f hatf}
\fl f(z)= \det \left( \begin{array}{cc}
    \A & \B \\
    \UI_B & \frac{1}{z} \, \tan (z\L) \\
    \end{array} \right), \qquad
\hat{f}(t)= \det \left( \begin{array}{cc}
    \A & \B \\
    \UI_B & \frac{1}{t} \, \tanh (t\L) \\
    \end{array} \right) \ ,
\end{equation}
so $\hat{f}(t)=f(\ui t)$.  First we note that
\begin{equation}\label{eq:star gc f(0)}
 f(0)= \hat{f}(0) =
    \det \left( \begin{array}{cc}
    \A & \B \\
    \UI_B & \L \\
    \end{array} \right) \ .
\end{equation}
 $f(0)$ is generically non-zero in the sense that if $f(0)=0$ perturbing the set of bond lengths in $\L$ by an arbitrarily small amount will make it non-zero.  Secondly we will also be concerned with the behavior of $\hat{f}$ in the limit $t$ to infinity.  In this limit $\hat{f}$ has an asymptotic expansion of the form
\begin{equation}\label{eq:star gen hatf asymp}
    \hat{f}(t) \sim \det \B +\frac{a_1}{t}+\frac{a_2}{t^2} +\dots \ .
\end{equation}
We denote by $a_N$ the first non-zero coefficient in the expansion, so $N=0$ if $\det \B \ne 0$.

We are now ready to state the following theorem for the zeta function of the star graph with a general matching condition at the central vertex.
\begin{theorem}\label{thm:star gc zeta}
For the Laplace operator on a star graph with general matching conditions at the center, defined by matrices $\A$ and $\B$ with $\A\B^\dagger=\B\A^\dagger$, $\rk (\A,\B)=B$, and Dirichlet boundary conditions at the nodes the zeta function on the strip $-1/2 < \re \, s <1$ is given by
\begin{eqnarray}
    \zeta(s)=& \zeta_H(2s,\frac{1}{2} ) \sum_{b=1}^B \left(\frac{\pi}{L_b}\right)^{-2s}
    -\frac{N\sin\pi s}{2\pi s} +\frac{\sin \pi s}{\pi} \int_0^1 t^{-2s} \frac{\ud}{\ud t} \log \hat{f}(t) \, \ud t \nn \\
    &+\frac{\sin \pi s}{\pi} \int_1^\infty t^{-2s} \frac{\ud}{\ud t} \log (t^N \hat{f}(t)) \, \ud t ,\nn \\
    \hat{f}(t)=& \det \left( \begin{array}{cc}
    \A & \B \\
    \UI_B & \frac{1}{t} \, \tanh (t\L) \\
    \end{array} \right) \ . \nn
\end{eqnarray}
\end{theorem}

A direct consequence of Theorem \ref{thm:star gc zeta} is the following simple formula for $\zeta'(0)$,
\begin{equation}\label{eq:star gc zeta'}
    \zeta'(0)=-\log \left( \frac{2^B}{a_N} \det \left( \begin{array}{cc}
    \A & \B \\
    \UI_B & \L \\
    \end{array} \right)
    \right).
\end{equation}
\begin{proof}[Proof of Theorem \ref{thm:star gc zeta}.]
Following the zeta function calculations for star graphs with a Neumann like matching condition at the center the theorem can be established in a few lines.
Again we consider the contour integral
\begin{equation}\label{eq:star gc contour int}
    \zeta(s)=\int_c z^{-2s} \frac{f'(z)}{f(z)}\,  \ud z = \int_c z^{-2s} \frac \ud {\ud z}  \log f(z)\,  \ud z
\end{equation}
with $c$ as shown in Figure \ref{fig:star contour} and $f$ defined in (\ref{eq:star gc f hatf}).
We again split the zeta function where
$\zeta_{Im}$ is the contribution of the integral on the imaginary axis and $\zeta_{P}$ is the contribution of the poles of $f$ on the positive real axis.

The poles of $f$ are at $(m+1/2)\pi L_b^{-1}$ for $m$ integer, the same
set of poles as the star with Neumann matching at the center and Neumann boundary conditions at the nodes, so
\begin{equation}\label{eq:star gc zeta P}
    \zeta_{P}(s) = \zeta_{H} (2s,1/2) \sum_{b=1}^B \left( \frac{\pi}{L_b} \right)^{-2s} \ .
\end{equation}
On the imaginary axis $f(\ui t)=\hat{f}(t)$ is an even function of $t$ and we have the representation
\begin{equation}
    \zeta_{Im}(s) = \frac{\sin \pi s}{\pi} \int_0^{\infty} t^{-2s} \frac{\ud}{\ud t} \log \hat{f}(t) \,  \ud t \ .
    \label{eq:star gc zeta Im}
\end{equation}
As $t$ tends to zero $\hat{f}(t)\sim \hat{f}(0) + c_1 t^2 +\dots$.  $\hat{f}(0)$ is generically non zero and consequently, for $N>0$
(\ref{eq:star gc zeta Im}) is valid in the strip $0<\re \, s<1$, whereas for $N=0$ it is valid in the strip $-1/2 < \re \, s <1$.  Splitting the integral at $t=1$ and subtracting the asymptotic behavior of $\hat{f}$ as $t$ tends to infinity, see (\ref{eq:star gen hatf asymp}), we obtain the theorem.
\end{proof}

\subsection{Comparison with results for a star with Neumann matching at the center.}
\label{sec:compare gereral to Neumann}
For completeness we calculate $\zeta'(0)$ using (\ref{eq:star gc zeta'}) when the central vertex has a Neumann matching condition in order to establish agreement with our previous results.
For Neumann like matching condition the matrices $\A$ and $\B$ can be chosen to be
\begin{equation}\label{eq:Neumann matching matrices}
    \A=\left( \begin{array}{cccccc}
    1&-1&0&\dots &0 \\
    0&1&-1& \ddots &   \vdots \\
    \vdots &\ddots & \ddots &\ddots & 0\\
    0& \dots & 0& 1 & -1 \\
    0& \dots & 0 &0  & 0 \\
    \end{array} \right) \ , \qquad
    \B= \left( \begin{array}{ccccc}
    0&0&\dots & 0 \\
    \vdots & \vdots &  &\vdots \\
    0& 0& \dots & 0 \\
    1& 1 & \dots  & 1 \\
    \end{array} \right) \ .
\end{equation}
Then equation (\ref{eq:matching conditions}) implies $\psi_b(L_b)=\psi_{b+1}(L_{b+1})$ for $b=1,\dots,B-1$ and $\sum_{b=1}^{B} \psi'_b(L_b)=0$.
To evaluate $\zeta'(0)$ using (\ref{eq:star gc zeta'}) we need to evaluate $\hat{f}(0)$ and $a_N$. We first note that
\begin{equation}\label{eq:general det at zero}
    \hat{f}(0) = \det \left( \begin{array}{cc}
    \A&\B\\
    \UI_B&\L\\
    \end{array} \right) =\det (\A\L-\B) .
\end{equation}
Given our choice of $\A$ and $\B$
\begin{eqnarray}
    \hat{f}(0)&=&\det \left(
    \begin{array}{cccccc}
    L_1&-L_2&0&\dots &0 \\
    0&L_2&-L_3&\ddots &  \vdots \\
    \vdots & \ddots & \ddots &\ddots &  0\\
    0& \dots & 0& L_{B-1} & -L_B \\
    -1& \dots & -1 &-1  & -1 \\
    \end{array} \right)  \\ &=& L_1 \det \left(
    \begin{array}{cccccc}
    L_2&-L_3&0&\dots &0 \\
    0&L_3&-L_4& \ddots &   \vdots\\
    \vdots & \ddots & \ddots &\ddots & 0 \\
    0& \dots & 0& L_{B-1} & -L_B \\
    -1& \dots & -1 &-1  & -1 \\
    \end{array} \right) -L_1^{-1}\prod_{b=1}^B L_b \ .
\end{eqnarray}
Iterating this procedure
\begin{equation}
\fl    \det (\A\L-\B)=L_1 L_2 \det \left(
    \begin{array}{cccccc}
    L_3&-L_4&0&\dots &0 \\
    0&L_4&-L_5&\ddots  & \vdots  \\
    \vdots & \ddots & \ddots &\ddots &0 \\
    0& \dots & 0& L_{B-1} & -L_B \\
    -1& \dots & -1 &-1  & -1 \\
    \end{array} \right) -L_1^{-1}\prod_{b=1}^B L_b - L_2^{-1} \prod_{b=1}^B L_b \ .
\end{equation}
Consequently, we find
\begin{equation}\label{eq:gen star hatf0}
    \hat{f}(0)=\det \left( \begin{array}{cc}
    \A&\B\\
    \UI_B&\L\\
    \end{array} \right) = -\left( \prod_{b=1}^B L_b \right) \left( \sum_{b=1}^B L_b^{-1} \right) \ .
\end{equation}

To obtain the $t$ to infinity behavior of $\hat{f}$ we rewrite the determinant so
\begin{equation}\label{eq:star gc comp hatf alt}
    \hat{f}(t) = \det \left( \begin{array}{cc}
    \A & \B \\
    \UI_B & \frac{1}{t} \, \tanh (t\L) \\
    \end{array} \right) = \frac{1}{t^{B-1}} \det
    \Big( \A \tanh (t\L) - \B \Big)
\end{equation}
as the rank of $\A$ is $B-1$.  As $t$ tends to infinity $\tanh (t\L)$ approaches $\UI_B$.
The first non-zero coefficient in the asymptotic expansion of $\hat{f} (t)$ is
\begin{equation}\label{eq:star gc aN}
    a_{B-1}=\det(\A-\B)=-B \ ,
\end{equation}
where the determinant can be evaluated following the same expansion used to evaluate $\hat{f}(0)$.
Substituting the values of $a_{B-1}$ and $\hat{f}(0)$ in
(\ref{eq:star gc zeta'}) we again obtain (\ref{eq:dn* zeta' 0}).

%
%

\section{The zeta function of a general quantum graph.}\label{sec:general zeta}
A secular equation based on a bond scattering matrix of the graph is a widely employed starting point for the derivation of the trace formula of a quantum graph, see e.g.
\cite{p:BE:TFQGGSABC, p:KS:POTSSQG}.  However, the approach we have adopted so far has a natural analogy with techniques used to study zeta functions associated to manifolds which we would like to maintain.  To achieve this it is necessary to employ an equivalent formulation for a secular equation associated to a general quantum graph in which the matching conditions appear directly.

\subsection{Secular equations for general quantum graphs.}
For comparison, the secular equation for a general quantum graph defined in terms of a bond scattering matrix (or quantum evolution operator) was introduced by Kottos and Smilansky \cite{p:KS:QCG, p:KS:POTSSQG}. In our notation the quantization condition on the graph takes the form
\begin{equation}\label{eq:secular traditional}
    \det \left( \UI - \SM (k) \left( \begin{array}{cc}
    0 & \ue^{\ui k \L} \\
    \ue^{\ui k \L} & 0 \\
    \end{array} \right) \right) = 0 \ .
\end{equation}
%
 %
 %
 The scattering matrix $\SM$ is a unitary $2B\times 2B$ matrix which can be defined in terms of the matching conditions on the graph \cite{p:KS:KRQW},
\begin{equation}\label{eq:scattering matrix}
    \SM(k)=-(\A+\ui k \B)^{-1}(\A-\ui k\B) \ .
\end{equation}

The following alternative formulations of the secular equation are all equivalent and equivalent to (\ref{eq:secular traditional}); they differ primarily in the form of the functions used to incorporate the $k$-dependence.  We first simply list the secular equations before indicating how to derive them.

The secular equation written in terms of the boundary conditions using exponential functions has the form
\begin{equation}\label{eq:secular exponential}
    \det \left( \A
    \left( \begin{array}{cc}
    \ue^{-\ui k \L/2} & \ue^{\ui k \L/2}\\
    \ue^{\ui k \L/2}&\ue^{-\ui k \L/2}\\
    \end{array} \right)+\ui k \B
    \left( \begin{array}{cc}
    \ue^{-\ui k \L/2} & -\ue^{\ui k \L/2}\\
    -\ue^{\ui k \L/2}&\ue^{-\ui k \L/2}\\
    \end{array} \right)
    \right)=0 \ .
\end{equation}
Alternatively a similar quantization condition can be obtained with trigonometric functions, namely
\begin{equation}\label{eq:secular trig 1}
\fl    \det \left( \A
    \left( \begin{array}{cc}
    -\sin( k \L/2)  & \cos( k \L/2)\\
    \sin (k \L/2)&\cos ( k \L/2) \\
    \end{array} \right)+ k \B
    \left( \begin{array}{cc}
    \cos (k \L/2) & \sin (k \L/2)\\
    -\cos (k \L/2)&\sin ( k \L/2)\\
    \end{array} \right)
    \right)=0 .
\end{equation}
%
%
%
The final formulation that we employ later has a similarity with the secular equation of the Neumann star graph
(\ref{eq:secular star}),
%
\begin{equation}\label{eq:secular trig 2}
    \det \left( \A
    + k \B
    \left( \begin{array}{cc}
    -\cot(k \L) & \csc(k \L)\\
    \csc(k \L)& - \cot( k \L)\\
    \end{array} \right)
    \right) =0 .
\end{equation}

Derivation of (\ref{eq:secular trig 1}): the matching conditions on the graph are specified by the matrix equation
\begin{equation}\label{eq:graph matching again}
    \A \vp +\B\vp'=\vz \ .
\end{equation}
On each bond the eigenfunction has the form
\begin{equation}\label{eq:bond wavefn}
    \psi_b(x_b)=c_b\sin kx_b +\hat{c}_b \cos kx_b \ .
\end{equation}
Consequently the vectors of values of the function and its derivative at the ends of the intervals corresponding to each bond are
\begin{eqnarray}
\fl  \vp &=& (\hat{c}_1,\dots,\hat{c}_B, c_1\sin kL_1+\hat{c}_1\cos kL_1,\dots,
  c_B\sin kL_B+\hat{c}_B\cos kL_B)^T \ ,\\
\fl  \vp' &=& k(c_1,\dots,c_B, \hat{c}_1\sin kL_1-c_1\cos kL_1,\dots,
  \hat{c}_B\sin kL_B-c_B\cos kL_B)^T \ .
\end{eqnarray}
Writing $\vc=(c_1,\dots,c_B)^T$ and $\hat{\vc}=(\hat{c}_1,\dots \hat{c}_B)^T$ the matching conditions (\ref{eq:graph matching again}) take the form
\begin{equation}\label{eq:derive sec 1}
  \fl  \A \left( \begin{array}{cc}
    0 & \UI \\
    \sin (k\L)& \cos (k\L) \\ \end{array} \right)
    \left( \begin{array}{c}
    \vc \\
    \hat{\vc} \\
    \end{array} \right) +k\B
    \left( \begin{array}{cc}
    \UI & 0 \\
    -\cos (k\L)& \sin (k\L) \\ \end{array} \right)
    \left( \begin{array}{c}
    \vc \\
    \hat{\vc} \\
    \end{array} \right)
    =\vz \ .
\end{equation}
Eigenfunctions therefore exist for values of $k$ that solve the secular equation,
\begin{equation}\label{eq:derive sec 2}
\fl    \det \left( \A \left( \begin{array}{cc}
    0 & \UI \\
    \sin (k\L)& \cos (k\L) \\ \end{array} \right)
    +k\B
    \left( \begin{array}{cc}
    \UI & 0 \\
    -\cos (k\L)& \sin (k\L) \\ \end{array} \right)
    \right)
    =0 \ .
\end{equation}
This can be reduced to the form given in (\ref{eq:secular trig 1}) post-multiplying by
%
%
%
\begin{equation*}\label{eq:det c-ssc}
    \left( \begin{array}{cc}
    \cos (k\L/2) & \sin (k\L/2) \\
    -\sin (k\L/2)& \cos (k\L/2) \\ \end{array} \right) \ ,
\end{equation*}
a matrix with unit determinant.
The secular equation (\ref{eq:secular exponential}) is obtained
%
%
from (\ref{eq:secular trig 1}) after post-multiplying by
\begin{equation*}
    \left( \begin{array}{cc}
    \ui \UI_B & -\ui \UI_B \\
     \UI_B & \UI_B \\
    \end{array} \right) \ .
\end{equation*}
Alternatively the same result comes from expressing the wavefunctions on the bonds as sums of plane waves.
%
%
Finally equation (\ref{eq:secular trig 2}) can be obtained from (\ref{eq:derive sec 2}) by multiplication on the right by
\begin{equation}\label{eq:proving secular trig 2}
    \left( \begin{array}{cc}
    0 & \UI_B \\
    \sin (k\L) & \cos (k\L) \\
    \end{array} \right)^{-1}
    = \left( \begin{array}{cc}
    -\cot(k\L) & \csc(k\L) \\
    \UI_B & 0 \\
    \end{array} \right) \ .
\end{equation}
This matrix has a nonzero determinant and the inverse exists provided $k$ is not a zero of $\sin kL_b$ for any $b$.

\subsection{General zeta function calculation.}\label{sec:gen zeta}
To derive the zeta function it is convenient to start with the secular equation written in the
form (\ref{eq:secular trig 2}) which resembles the equation (\ref{eq:secular star}) for a star graph that we began with.
In the general graph case we define the two functions
\begin{eqnarray}\label{eq:gen f}
    f(z)= \det \left( \A
    - z \B
    \left( \begin{array}{cc}
    \cot(z \L) & -\csc(z \L)\\
    -\csc(z \L)& \cot(z \L)\\
    \end{array} \right)
    \right) , \\
    \hat{f}(t)= \det \left( \A
    - t \B
    \left( \begin{array}{cc}
    \coth(t \L) & -\csch (t \L)\\
    -\csch (t \L)& \coth(t \L)\\
    \end{array} \right)
    \right) \ , \label{eq:gen fhat}
\end{eqnarray}
such that $\hat{f}(t)=f(\ui t)$. Poles of $f$ lie at integer multiples of $\pi /L_b$ on the real axis.  When the set of bond lengths $\{ L_b \}$ is  incommensurate the poles of $\cot{zL_b}$ are all distinct.  However, as there are four functions with the same set of poles in the determinant it is possible for a linear combination of the functions to cancel some of the poles of $\cot$.  Generically, however, the set of poles is $\{ m\pi/L_b | m\in \gz, b=1,\dots,B \}$ and this is the case we consider.

Again we represent $\zeta (s)$ using the contour $c$, see Figure \ref{fig:star contour}.
\begin{equation}\label{eq:gen contour int}
    \zeta(s)=\int_c z^{-2s} \frac{f'(z)}{f(z)}\,  \ud z = \int_c z^{-2s} \frac \ud {\ud z}\log  f(z)  \ud z .
\end{equation}
Notice, that at the origin
\begin{equation}\label{eq:gen f(0)}
    f(0)=
    \det \left( \A
    - \B
    \left( \begin{array}{cc}
    \L^{-1} & -\L^{-1}\\
    -\L^{-1} & \L^{-1}\\
    \end{array} \right)
    \right) = \hat{f}(0) \ .
\end{equation}
If $f(0)=0$ an arbitrarily small change in the bond lengths can make this non-zero, so the integral formulation naturally excludes zero modes.

As before, we split the zeta function into the sum of the imaginary axis integral and the contribution of the poles on the real axis $\zeta(s)=\zeta_{Im}(s)+\zeta_{P}(s)$.
As we have previously seen in the case of the star graph for this set of poles
\begin{equation}\label{eq:dn* zeta P2}
    \zeta_{P}(s) = \zeta_{R}(2s) \sum_{b=1}^{B} \Big( \frac{\pi}{L_b} \Big)^{-2s} \ .
\end{equation}

On the imaginary axis
\begin{equation}
    \zeta_{Im}(s) = \frac{\sin \pi s}{\pi} \int_0^{\infty} t^{-2s} \frac{\ud}{\ud t} \log \hat{f}(t) \,  \ud t  \ .
    \label{eq:star gc zeta Im2}
\end{equation}
As $t$ tends to infinity, up to exponentially damped terms,
\begin{equation}\label{eq:gen fhat infinity}
    \hat{f}(t) \sim
    \det \left( \A
    - t \B
    \right) = \det \B \, t^{2B} + c_{2B-1} t^{2B-1} + \dots + c_{1} t +\det \A
\end{equation}
Let $c_N$ be the first non-zero coefficient of the highest power of $t$, i.e. $N=2B$ when $\det \B\ne 0$.  Splitting the integral at $t=1$ and subtracting the $t$ to infinity behavior we obtain the following integral formulation of the zeta function.

\begin{theorem}\label{thm:gen zeta}
For the Laplace operator on a graph whose vertex matching conditions are defined by a pair of matrices $\A$ and $\B$, with $\A\B^\dagger=\B\A^\dagger$ and $\rk (\A,\B)=B$, the zeta function on the strip $-1/2 < \re \, s <1$ is given by
\begin{eqnarray}
    \zeta(s)=& \zeta_{R}(2s) \sum_{b=1}^{B} \Big( \frac{\pi}{L_b} \Big)^{-2s}
    +\frac{N\sin\pi s}{2\pi s} +\frac{\sin \pi s}{\pi} \int_0^1 t^{-2s} \frac{\ud}{\ud t} \log \hat{f}(t) \, \ud t \nn \\
    &+\frac{\sin \pi s}{\pi} \int_1^\infty t^{-2s} \frac{\ud}{\ud t} \log (t^{-N} \hat{f}(t)) \, \ud t \ , \nn \\
    \hat{f}(t)=& \det \left( \A
    - t \B
    \left( \begin{array}{cc}
    \coth(t \L) & -\csch (t \L)\\
    -\csch (t \L)&  \coth(t \L)\\
    \end{array} \right)
    \right) \  , \nn
\end{eqnarray}
where the poles of $f(z)=\hat{f}(-\ui z)$ are the whole of the set
$\{ m\pi/L_b | m\in \gz, b=1,\dots,B \}$.
\end{theorem}

Theorem \ref{thm:gen zeta} implies the following simple formula for the derivative of zeta at zero in terms of the matching conditions at the graph vertices.
\begin{equation}\label{eq:gen zeta'0}
    \zeta'(0)=-\log \left[ \frac{2^B}{c_N \prod_{b=1}^B L_b} \det
    \left( \A - \B
    \left( \begin{array}{cc}
    \L^{-1} & -\L^{-1}\\
    -\L^{-1} & \L^{-1}\\
    \end{array} \right)
    \right)\right].
\end{equation}

We have now obtained results for the spectral zeta function of a generic quantum graph along with a number of specific zeta function calculations for the formative case of a star.  In the balance we assess the consequences of these results for the spectral determinant, vacuum energy and heat kernel asymptotics of quantum graphs.

\section{Comparison with results for the spectral determinant.}\label{sec:spectral det}

The spectral determinant of a Schr\"odinger operator on a quantum graph is an important subject in its own right which has previously been investigated by a number of authors \cite{p:D:SDSOG,p:D:SDGGBC,p:F:DSOMG}.  Formally the spectral determinant of the graph Laplacian is
\begin{equation}\label{eq:spec det formal}
    {\det}' (-\triangle)=\prod_{j=0}^\infty \phantom{|}^{\prime} \lambda_j \ .
\end{equation}
The zeta function representation presented here allows a direct evaluation of the regularized spectral determinant, ${\det}' (-\triangle)=\exp(-\zeta'(0))$.  Consequently, (\ref{eq:gen zeta'0}) provides the spectral determinant of the Laplace operator on a general graph in terms of the vertex matching conditions presented in Theorem \ref{thm:gen spec det}.

We wish to compare our results derived from the zeta function with the literature.  As the known formulations of the spectral determinant each apply to certain classes of quantum graphs for the purpose of comparison we take the example of the Laplace operator on a star with Neumann conditions at all the vertices.  From the star graph zeta function, see (\ref{eq:zeta' 0}), in this case
\begin{equation}\label{eq:spec det N star}
    {\det}' (-\triangle)=\frac{2^B \cL}{B} \ .
\end{equation}

Friedlanders formulation for the spectral determinant of Schr\"odinger operators on graphs with delta type vertex matching conditions \cite{p:F:DSOMG} reduces, in the case of Neummann like conditions at the vertices, to the following spectral determinant.
\begin{equation}\label{eq:Friedlander}
    {\det}' (-\triangle)=2^{B} \frac{\cL}{V}\frac{\prod_b L_b}{\prod_v d(v)} \, {\det}' R .
\end{equation}
The notation ${\det}'$ denotes the determinant excluding eigenvalues of zero and $d(v)$
is the degree of the vertex $v$.  $R$ is the Dirichlet-to-Neumann operator at zero energy, a $V\times V$ matrix defined by
\begin{equation}\label{eq:defn R}
    R_{vw} = \left\{ \begin{array}{cl}
    -\sum_{b\in [v,w]} L_b^{-1} & \quad v\ne w \\
    \sum_{b \sim v } L_b^{-1} & \quad v= w \\
    \end{array}
    \right. \ .
\end{equation}
When $v\ne w$ the sum in (\ref{eq:defn R}) is over bonds connecting $v$ and $w$ and when $v=w$ one sums
over bonds originating at $v$ excluding loops.    For the star graph $V=B+1$ and the matrix $R$ can be written in the form
\begin{equation}
R=\left( \begin{array}{ccccc}
\sum_{b=1}^B L_b^{-1} & -L_1^{-1} & -L_2^{-1} & \dots & -L_B^{-1} \\
-L_1^{-1} & L_1^{-1} & 0 & \dots & 0 \\
-L_2^{-1} & 0 & L_2^{-1} & \ddots & \vdots  \\
\vdots & \vdots & \ddots & \ddots & 0  \\
-L_B^{-1} & 0 & \dots & 0 & L_B^{-1} \\
\end{array}
\right) \ .
\end{equation}
It is convenient to combine $\prod_b L_b$ with ${\det}'R$ so
\begin{equation}
    \left( \prod_{b=1}^{B} L_b \right) {\det}'R = {\det}'
    \left( \begin{array}{ccccc}
    \sum_{b=1}^B L_b^{-1} & -L_1^{-1} & -L_2^{-1} & \dots & -L_B^{-1} \\
    -1 & 1 & 0 & \dots & 0 \\
    -1 & 0 & 1 & \ddots & \vdots \\
    \vdots & \vdots & \ddots & \ddots & 0 \\
    -1 & 0 & \dots & 0 & 1 \\
    \end{array}
    \right)   \ .
\end{equation}
Applying column operations it is clear that this matrix has eigenvalues $0$ and $1$ with multiplicity $1$ and $B$ respectively. Consequently,
for the star graph $(\prod_b L_b) \, {\det}'R=1$ and
\begin{equation}\label{eq:Friedlander star}
    {\det}' (-\triangle)=2^{B} \frac{\cL}{VB} \ .
\end{equation}
%
%
This does not agree with (\ref{eq:spec det N star}), the difference between the results being the factor $V$. Restricting to the interval with Neumann boundary conditions, the associated zeta function is
\begin{equation}
    \zeta (s) = \left( \frac \pi L \right)^{-2s} \zeta_R (2s)
\end{equation}
and
\begin{equation}
    \zeta ' (0) = - \log (2L).
\end{equation}
This situation corresponds to $B=1$, $\cL =L$
and $V=2$.  Consequently (\ref{eq:spec det N star}) is seen to be correct differing from (\ref{eq:Friedlander star}) by some normalization constant.
Finally, it is worth noting that Theorem \ref{thm:gen spec det} agrees with the conjecture of Texier for the spectral determinant of a quantum graph with general boundary conditions \cite{p:T:ZRSDMG}.


\section{Vacuum energy and Casimir force of graphs.}\label{sec:vacuum energy}

The vacuum (Casimir) energy of quantum graphs has also been a topic of recent research \cite{p:BHW:MAVEQG, p:FKW:VERCFQSG}.
A quantum graph provides a simple model of vacuum energy for a system whose corresponding classical dynamics is ergodic.  The quantum graph model also naturally introduces a set of independent length scales, the set of bond lengths, so techniques applicable to quantum graphs may provide insight into the approaches required to tackle less symmetric but physically relevant structures in thee dimensions. However, the quantum graph vacuum energy should only be thought of as a mathematical construction rather than a lower dimensional model for the vacuum energy fluctuations of some underlying physical configuration of narrow tubes.  It remains to be established how the vacuum energy of a graph is related to the vacuum energy of a thickened manifold with specific boundary conditions built over the graph skeleton.

Formally, the vacuum energy associated with the Laplace operator on a graph is given by
\begin{equation}\label{eq:E_c}
    E_c=\frac{1}{2} \sum_{j=0}^{\infty} \phantom{|}^{\prime} \sqrt{\lambda_j} \ .
\end{equation}
The zeta function regularization of this sum is $E_c=\frac{1}{2} \zeta ( -1/2 )$.
However, not all the zeta function expressions constructed so far are valid for
$\re \, s =-1/2$ and these cases will need to be examined with more care. As we will see, $\zeta(-1/2)$ generically diverges, but the Casimir force acting on a bond will be well defined.

The vacuum energy of a star graph with equal bond lengths was investigated by Fulling, Kaplan and Wilson in \cite{p:FKW:VERCFQSG} using the heat kernel regularization.  Of particular interest they discovered that the sign of the Casimir force depends on the number of bonds and for star graphs with more than three bonds the force is repulsive.
In \cite{p:BHW:MAVEQG} Berkolaiko, Harrison and Wilson applied the quantum graph trace formula to obtain periodic orbit expansions of the vacuum energy of general graphs.  The zeta function approach is new and complimentary, relating the graph vacuum energy directly to the matching conditions at the graph vertices.

\subsection{Star graphs with equal bond lengths.}
We begin with the simple case of a star with Neumann matching conditions at the center and Neumann boundary conditions at the nodes. The zeta function for equal bond lengths $L$ was given in equation (\ref{eq:n* L zeta}).  Evaluating the vacuum energy
\begin{eqnarray}
  E_c &=& \frac{\pi}{4L} (3-B) \zeta_{R}(-1)  \nn \\
   &=& \frac{\pi}{48L} (B-3) \label{eq:n* L E_c}
\end{eqnarray}
which agrees with the results in \cite{p:BHW:MAVEQG, p:FKW:VERCFQSG}.  The Casimir force is proportional to the derivative of $E_c$ with respect to the bond length $L$ and we see that for $B>3$ increasing the bond lengths is energetically favorable and consequently the Casimir force is repulsive.

If we now allow Dirichlet boundary conditions at some nodes, $\zeta(s)$ was given in equation (\ref{eq:dn* zeta L}), and we find
\begin{eqnarray}
  E_c &=& \frac{\pi}{4L} \left[ 2\Big( \zeta_H (-1,\alpha) + \zeta_{H} (-1,1-\alpha)
  \Big)  +(2B_D-B_N-1)\zeta_R(-1) \right] \nn \\
   &=& \frac{\pi}{48L}(B_N - 2B_D +1) +\frac{\pi}{2L} \Big( \zeta_H(-1,\alpha) + \zeta_{H} (-1,1-\alpha) \Big).
    \label{eq:dn* L E_c part}
\end{eqnarray}
Using the relationship between the Hurwitz zeta function at negative integers and Bernoulli polynomials one can obtain,
\begin{equation}\label{eq:gen zeta prop}
    \zeta_{H}(-1,\alpha) + \zeta_{H} (-1,1-\alpha)=-\frac{1}{6}+\alpha-\alpha^2 \ ,
\end{equation}
from which we can write the vacuum energy in a concise form as a function of $B$ and $B_D$, the number of nodes with Dirichlet boundary conditions,
\begin{equation}\label{eq:dn* L E_c}
    E_c=\frac{\pi}{48L}\Big( B - 3(B_D+1) +24\alpha(1-\alpha) \Big) .
\end{equation}
 Equation (\ref{eq:dn* L E_c}) generalizes the previous results for star graphs.  For example on the graph with three bonds the vacuum energy is positive if $B_D=1$ and negative for $B_D=2$.  In general increasing the number of Dirichlet nodes changes the vacuum energy from positive to negative and consequently changes the Casimir force from repulsive to attractive.

\subsection{Star graphs with incommensurate bond lengths.}
For star graphs where the set of bond lengths is incommensurate the presentations of the zeta function provide an integral formulation of the vacuum energy in contrast to the conditionally convergent periodic orbit sum obtained by Berkolaiko et al. \cite{p:BHW:MAVEQG}.
For a star with a Neumann matching condition at the center and mixed Dirichlet and Neumann boundary conditions at the nodes from Theorem \ref{thm:zeta mixed star} we obtain the vacuum energy
\begin{equation}\label{eq:dn* E_c}
 \fl   E_c=\frac{\pi}{48} \Big( \sum_n L_n^{-1} - 2\sum_d L_d^{-1} \Big)
    -\frac{1}{2\pi} \int_0^{\infty} t \frac{\sum_n L_n \textrm{sech}^2 (tL_n)-\sum_d L_d \textrm{csch}^2(t L_d)}{\sum_n \tanh t L_n +\sum_d \coth t L_d} \, \ud t \ .
\end{equation}
%
A similar formula holds when all the nodes have Neumann boundary conditions.

For comparison, if we set the bond lengths equal, (\ref{eq:dn* E_c}) reduces to
\begin{equation}\label{eq:dn* E_c reduced}
 \fl   E_c=\frac{\pi}{48L} \Big( B - 3B_D \Big)
    -\frac{L}{2\pi} \int_0^{\infty} t \frac{B_N \textrm{sech}^2 (tL)-B_D \textrm{csch}^2 (tL) }{B_N \tanh (tL)+B_D \coth (tL)} \, \ud t \ .
\end{equation}
This is an integral formulation of the vacuum energy evaluated previously
(\ref{eq:dn* L E_c}).  As an indirect consequence we have obtained the following integral
\begin{equation}\label{eq:new integral}
\fl \int_0^{\infty} x \frac{a\, \textrm{sech}^2 x- b\,\textrm{csch}^2 x }{a\, \tanh x+ b\,\coth x} \, \ud x = \sin^{-1}\sqrt{\frac{b}{a+b}} \left[ \sin^{-1} \sqrt{\frac{b}{a+b}} - \pi \right] + \frac{\pi^2}{8} \ .
\end{equation}
This does not appear to have been known.
%

\subsection{Casimir force on a star with a general matching condition at the center.}
\label{sec:star gc vac}
The integral formulation of the graph zeta function obtained in Section \ref{sec:general star} was for values of $s$ in the strip $-1/2<\re \, s<1$.  To analyze the vacuum energy we must first continue the zeta function formula to include $s=-1/2$.  The restriction
$-1/2<\re \, s$ came from the behavior of the function $\hat{f}$ at infinity.
\begin{equation}
\hat{f}(t)= \det \left( \begin{array}{cc}
    \A & \B \\
    \UI_B & \frac{1}{t} \, \tanh (t\L) \\
    \end{array} \right) \sim \frac{a_N}{t^N}+\frac{a_{N+j}}{t^{N+j}}+O(t^{-(N+j+1)}),
\end{equation}
where $a_N$ was the first nonzero term in the asymptotic expansion of $\hat{f}$ and $a_{N+j}$ is the second nonzero coefficient, generically $j=1$.
From this we find
\begin{equation}
\log \hat{f}(t)\sim \log \left( \frac{a_N}{t^N} \right) +\frac{a_{N+j}}{a_N t^j}+O(t^{-(j+1)}).
\end{equation}
Recall
\begin{equation}\label{eq:Fc Im int}
    \zeta_{Im}(s) = \frac{\sin \pi s}{\pi} \int_0^{\infty} t^{-2s} \frac{\ud}{\ud t} \log \hat{f}(t) \,  \ud t \ .
\end{equation}
When we develop the imaginary axis integral subtracting the leading and subleading order behavior at infinity leaves an integral convergent for $\re \, s > -(j+1)/2$ as required.
We arrive at the zeta function formula
\begin{eqnarray}
 \fl   \zeta(s)=& \zeta_H\Big( 2s,\frac{1}{2} \Big) \pi^{-2s}\sum_{b=1}^B L_b^{2s} +
    \frac{\sin \pi s}{\pi} \left[ \int_0^1 t^{-2s} \frac{\ud}{\ud t} \log \hat{f}(t) \, \ud t \right. \nn \\
 \fl   &\left. +\int_1^\infty t^{-2s} \frac{\ud}{\ud t} \left( \log \Big( t^N \hat{f}(t) \Big) -\frac{a_{N+j}}{a_N\, t^j} \right) \, \ud t -\frac{N}{2s} -\frac{a_{N+j}\, j}{a_N \, (2s+j)} \right]\ .
\end{eqnarray}
The vacuum energy $E_c$ is divergent if $j=1$ which would be the general  case. However, as terms dependent of the bond lengths are exponentially suppressed in the $t$ to infinity behavior both $a_N$ and $a_{N+j}$ are independent of the set of bond lengths
$\{ L_b \}$.  Changes in the vacuum energy and the corresponding Casimir forces are the observable quantities.  To find the Casimir force on bond $\beta$ the
vacuum energy is differentiated with respect to the bond length $L_\beta$ and therefore the Casimir force is well defined, namely
\begin{equation}\label{eq:Fc gstar}
   F_{c}^\beta= -\frac{\pi}{48\, L_\beta^2}
    +\frac{1}{2\pi} \int_0^\infty \frac{\partial }{\partial L_\beta }
    \log \hat{f}(t) \, \ud t .
\end{equation}
It is simple to differentiate $\hat{f}$ with respect to $L_\beta$ as only one term in the matrix determinant
depends on the length of $\beta$.

%
%
%
%
Although the form of the Casimir force looks somewhat different from our previous results (\ref{eq:dn* L E_c}) and (\ref{eq:dn* E_c}) we can compare them in the case when all the bond lengths are equal.
Inserting the Neumann matching matrices (\ref{eq:Neumann matching matrices}), we observe
\begin{equation}\label{eq:Fc nstar fhat}
    \hat{f}(t)=-\frac{1}{t^{B-1}} \sum_{b=1}^B \prod_{j =B, j\ne b} \tanh tL_j \ .
\end{equation}
Consequently $N=B-1$ and $a_{N+j}=0$.
Differentiating with respect to $L_\beta$ and setting $L_b=L$ for all $b\in \mathcal{B}$ we find
\begin{equation}\label{eq:Fc nstar dL}
    \frac{\partial}{\partial L_\beta} \log \hat{f} (t)=\frac{B-1}{B} \, \frac{t \, \sech^2 (tL)}{\tanh (tL)} \ .
\end{equation}
Substituting in (\ref{eq:Fc gstar}) and integrating we obtain a Casimir force
\begin{equation}\label{eq:Fc nstar dn}
    F_c=-\frac{\pi}{48\, B\, L^2} (3-2B) \ .
\end{equation}
The vacuum energy of a star with a Neumann center and Dirichlet nodes (\ref{eq:dn* L E_c}) when $B_D=B$ and consequently $\alpha=\pi^{-1}\arcsin \sqrt{B_D/B}=1/2$ is $E_c=\pi (3-2B)/48L$.
Differentiating, the Casimir force agrees with (\ref{eq:Fc nstar dn}) which is the force on the bond
$\beta$ rather than the force on all $B$ bonds of the star.

\subsection{Casimir force on a general quantum graph.}\label{sec:gen Ec}
We conclude with the equivalent result for the Casimir force on a bond of a general graph in terms of the matching conditions defined on the graph.  Here again we first continue the zeta function formula to a form valid at $s=-1/2$.  The vacuum energy will still generically be divergent as is the case with general matching conditions in a star.  However, the asymptotic behavior of $\hat{f}$ is independent of the bond lengths as there is no potential or curvature on the graph and consequently divergent terms do not appear in the Casimir force.

To obtain a form of the zeta function valid for $s=-1/2$ we again subtract both the leading and
subleading order behavior of $\hat{f}(t)$ in the integral along the imaginary axis, which follow from
\begin{eqnarray}
  \hat{f}(t)&= \det \left( \A
    - t \B
    \left( \begin{array}{cc}
    \coth(t \L) & -\csch (t \L)\\
    -\csch (t \L)& \coth(t \L)\\
    \end{array} \right)
    \right)  \nn \\
    & \sim \det (\A-t\B)= c_N t^N + c_{N-j}\, t^{N-j} + O(t^{N-j-1})  , 
\end{eqnarray}
where $c_N$ and $c_{N-j}$ are the highest order non-zero coefficients in the expansion of $\det (\A-t\B)$. From here, we continue noting
\begin{equation}
\log \hat{f}(t)\sim \log ( c_N \, t^N ) +\frac{c_{N-j}}{c_N t^j}+O(t^{-(j+1)}) \ .
\end{equation}
Subtracting the leading and subleading order asymptotic behavior from $\hat{f}$ in
(\ref{eq:Fc Im int}) leaves an integral convergent for $\re \, s>-(j+1)/2$.  The zeta function has the form
\begin{eqnarray}
 \fl   \zeta(s)=& \frac{\zeta_R(2s)}{\pi^{2s}}\sum_{b=1}^B L_b^{2s} +
    \frac{\sin \pi s}{\pi} \left[ \int_0^1 t^{-2s} \frac{\ud}{\ud t} \log \hat{f}(t) \, \ud t \right. \nn \\
\fl    &\left. +\int_1^\infty t^{-2s} \frac{\ud}{\ud t} \left( \log \Big( t^{-N} \hat{f}(t) \Big) -\frac{c_{N-j}}{c_N\, t^j} \right) \, \ud t +\frac{N}{2s} -\frac{c_{N-j}\, j}{c_N \, (2s+j)} \right] \ .
\end{eqnarray}
In general $j=1$ and $E_c=\frac{1}{2}\zeta(-1/2)$ is divergent.  The coefficients $c_n$ are independent of the bond lengths.  Differentiating with respect to $L_\beta$ and setting $s=-1/2$ we obtain the Casimir force on the bond $\beta$ stated in Theorem \ref{thm:Fc graph}.
%
%
%

\section{Heat kernel asymptotics.}\label{sec:heat kernel}
The heat kernel $K(t)$ has already been mentioned in the context of an alternative renormalization scheme for the vacuum energy. It is defined by
\begin{equation}\label{eq:heat kernel}
    K(t)=\sum_{j=1}^{\infty} \ue^{-\lambda_j t}
\end{equation}
and in the one-dimensional setting considered here it is known to have an asymptotic expansion as $t\to 0$ of the form
\begin{equation}\label{eq:heat kernel expansion}
    K(t) \sim \sum_{\ell =0,1/2,1,\dots}^{\infty} \varepsilon_\ell t^{\ell-1/2} \ .
\end{equation}
The heat kernel coefficients $a_\ell$ are related to the zeta function through the following connection \cite{seel68-10-288},
\begin{equation}\label{eq:heat coeffs}
    \varepsilon_\ell = \textrm{Res} (\zeta(s) \Gamma(s) )|_{s=1/2-\ell} \ .
\end{equation}
If we consider the heat kernel of a star graph with a Neumann matching condition at the center the zeta function representation can be used to evaluate the heat kernel coefficients. We find
\begin{equation}\label{eq:* heat kernel}
    K(t) \sim \left\{ \begin{array}{lcl} \case{\cL}{\sqrt{4\pi t}} -\frac{1}{2} && \textrm{Neuman nodes,} \\
    \frac{\cL}{\sqrt{4\pi t}} -\case{1}{2}(B_D-1) && \textrm{Dirichlet and Neuman nodes.} \\
    \end{array} \right.
\end{equation}
These results are independent of the incommensurability of the bond lengths.

When considering the asymptotic $t\to 0$ expansion of the heat kernel for the other cases we have to realize that the representations of the zeta functions are only valid for $-1/2 < \Re s <1$; see Theorems \ref{thm:star gc zeta} and \ref{thm:gen zeta}. In order to find the heat kernel coefficients for the general case we have to provide an analytic continuation that is valid further to the left of that strip.

 In detail, for the star graph with general matching conditions at the center, as shown previously, we have
\begin{eqnarray}
    \zeta(s)=& \zeta_H\left(2s,\frac{1}{2} \right) \sum_{b=1}^B \left(\frac{\pi}{L_b}\right)^{-2s}
    +\frac{\sin \pi s}{\pi} \int_0^1 t^{-2s} \frac{\ud}{\ud t} \log \hat{f}(t) \, \ud t \nn \\
    &+\frac{\sin \pi s}{\pi} \int_1^\infty t^{-2s} \frac{\ud}{\ud t} \log ( \hat{f}(t)) \, \ud t ,\nn \end{eqnarray}
with
\begin{eqnarray}
    \hat{f}(t)=& \det \left( \begin{array}{cc}
    \A & \B \\
    \UI_B & \frac{1}{t} \, \tanh (t\L) \\
    \end{array} \right) \ . \nn
\end{eqnarray}
The restriction $-1/2 < \Re s$ comes from the last integral and it is the $t\to\infty$ behavior of $\hat f (t)$, namely
\beq
\hat f (t) \sim \det \B + \frac{a_1} t + \frac{a_2} {t^2} + ..., \label{fhatasym}\eeq
that produces the restriction. By adding and subtracting the $t\to\infty$ behavior of the integrand in the last integral the analytical continuation can be found.

Consider the case $\det \B \neq 0$ first, then we have the structure \beq \log \hat f (t) \sim \log (\det \B ) + \sum_{n=1}^\infty \frac{b_n} {t^n} , \label{logfhat}\eeq
which defines the numerical multipliers $b_n$ using equation (\ref{fhatasym}). So furthermore
\beq \frac \ud {\ud t} \log \hat f (t) \sim - \sum_{n=1}^\infty \frac{nb_n} {t^{n+1}} .\nn\eeq
Adding and subtracting the $N$ leading asymptotic terms of this expansion we write
\beqa \fl \zeta (s) = \zeta_H \left( 2s , \frac 1 2 \right) \sum_{b=1}^B \left( \frac \pi {L_b} \right)^{-2s}
+\frac{\sin (\pi s)} \pi \int\limits_0^1 t^{-2s} \frac \ud {\ud t} \log \hat f (t) \ud t \label{heatana}\\
 +\frac {\sin (\pi s)} \pi \int\limits_1^\infty t^{-2s} \left[ \frac \ud {\ud t} \log \hat f (t) + \sum_{n=1}^N \frac{nb_{n}}{t^{n+1}}\right] \ud t
-\frac{\sin (\pi s)} \pi \sum_{n=1}^N \frac{nb_n}{2s+n},\nn\eeqa
a representation valid for $-(N+1)/2 < \Re s < 1$. This is the form we use, together with equation (\ref{eq:heat coeffs}), to find the heat kernel coefficients. For $s=1/2$ only the first term contributes to the residue and thus $a_0$. All other heat kernel coefficients are determined, by construction, by the last term.

From equation (\ref{heatana}) we read off \beqa \mbox{Res }\zeta \left( \frac 1 2 \right) &=& \frac{{\cal L}}{2\pi}, \quad \quad \zeta (0) =0, \nn\\
 \mbox{Res } \zeta \left( - \frac{2\ell +1} 2 \right) &=& (-1)^\ell \frac{ \left( \ell + \frac 1 2 \right) b_{2\ell +1}} \pi  , \quad \ell \in \N , \nn\\
 \zeta (-n) &=& (-1)^{n+1} n b_{2n} , \quad n\in\N.\nn\eeqa
 For the heat kernel coefficients this shows
 \beqa \varepsilon_0 &=& \frac{{\cal L}} {2\sqrt \pi} , \quad a_{1/2} =0, \nn\\
 \varepsilon_{\ell +1} &=& (-1)^\ell \frac{ \left( \ell + \frac 1 2 \right) \Gamma \left( - \ell - \frac 1 2 \right) b_{2 \ell +1} }  \pi =
 - \frac{b_{2 \ell +1}} {\Gamma \left( \ell + \frac 1 2 \right)}, \nn\\
 \varepsilon_{n+ \frac 1 2 } &=& - \frac{b_{2n}}{\Gamma (n)} . \nn\eeqa
 The full asymptotic heat kernel expansion therefore reads
 \beqa K(t) \sim \frac{{\cal L}} {\sqrt{4\pi t}} - \sum_{k=1,3/2,2,...} ^\infty \frac{b_{2k-1}}{\Gamma \left( k-\frac 1 2\right) } t^{k-1/2}.\label{detBnez}\eeqa
 Once the matching conditions are fixed the numerical coefficients  $b_{2k-1}$ are easily found from equations (\ref{fhatasym}) and (\ref{logfhat}) using an algebraic computer program.

 If $\det \B =0$, assume $a_N$ is the first non-vanishing coefficient in equation (\ref{fhatasym}). The relevant expansions in that case are \beqa \hat f (t) &\sim & \sum_{n=N}^\infty a_n t^{-n}, \nn\\
 \log \hat f (t) & \sim & - N \log t + \log c_N + \log \left( 1+\sum_{\ell =1}^\infty \frac{a_{\ell +N}}{a_N} t^{-\ell} \right) \nn\\
 &=& -N \log t + \log a_N + \sum_{n=1}^\infty b_n t^{-n} , \nn\eeqa
 where the numerical multipliers $b_n$ are defined by the last equation. Therefore,
 \beqa \frac \ud {\ud t} \log \hat f (t) \sim -\frac N t - \sum_{n=1}^\infty \frac{nb_n} {t^{n+1}} . \label{newlogfhat}\eeqa
Comparing with equation (\ref{logfhat}) the only difference is the additional first term which changes $\zeta (0) =\varepsilon_{1/2}=0$ into
$\zeta (0) =-N/2$. The heat kernel expansion now reads
\beqa
K(t) \sim \frac{{\cal L}} {\sqrt{4\pi t}}-\frac N 2  - \sum_{k=1,3/2,2,...} ^\infty \frac{b_{2k-1}}{\Gamma \left( k-\frac 1 2\right) } t^{k-1/2},\label{detBez}\eeqa
 so that $N=0$ reduces to equation (\ref{detBnez}) as it must. As before, once the matching conditions are fixed the numbers $b_{2k-1}$ are easily found.

Exactly the same calculation goes through for the general graph with general boundary conditions. The relevant $t\to\infty$ behavior this time is given by equation (\ref{eq:gen fhat infinity}). Let $N$ denote the highest power in this expansion; the structure then is
\beqa \hat f (t) &\sim  &\sum_{n=0}^N c_n t^n = t^N \sum_{\ell =0} ^N \frac{c_{N-\ell}} {t^\ell} , \nn\\
\log \hat f (t) & \sim & N \log t + \log c_N + \sum_{n=1} ^\infty \frac{b_n} {t^n}, \nn\eeqa
where the last equation defines the multipliers $b_n$. This shows
$$\frac \ud {\ud t} \log \hat f (t) \sim \frac N t - \sum_{n=1}^\infty \frac{nb_n} {t^{n+1}} , $$ and from the previous calculation we can immediately write down the heat kernel expansion
\beqa
K(t) \sim \frac{{\cal L}} {\sqrt{4\pi t}} +\frac N 2  - \sum_{k=1,3/2,2,...} ^\infty \frac{b_{2k-1}}{\Gamma \left( k-\frac 1 2\right) } t^{k-1/2}.\label{heatasygen}\eeqa

\section{The piston graph}\label{sec:piston}
To demonstrate how the techniques introduced in this article come together in a specific example of interest we apply the previous results to the case of a star graph with two bonds where there is a delta-type matching at the central vertex and Dirichlet conditions at the nodes.  The first bond has length $L_1=L$ and the remaining bond has length $L_2=\cL-L$.  One can therefore consider the graph as related to a narrow piston in the limit that the width shrinks to zero, see Figure \ref{fig:piston}.  Certainly in this limit the spectrum of an epsilon thick piston with appropriate matching condition on the piston membrane can be made to approach that of the corresponding graph.  In general quantum graph vertex matching conditions are approximable through a suitable choice of a Schr\"odinger operator on an epsilon thick neighborhood \cite{p:EP:AQGVCSSOTBM}.  Consequently we refer to this example as a piston graph where we regard the vacuum energy, and related Casimir force, as a function of the bond length $L$ keeping the total length of the piston $\cL$ fixed.  The piston membrane, the central vertex, may then move on the graph approaching one of the nodes keeping $\cL$ constant.  The piston model of vacuum energy was introduced by Cavalcanti \cite{p:C:CFP}, for recent developments see
\cite{p:FKKLM:VSCP, p:HJKS:ACFCG, p:HJKS:CFPGZFT,kifu, p:S:NCFGCP,teo}.

\begin{figure}[htb]
\begin{center}
\setlength{\unitlength}{9cm}
\begin{picture}(1.1,0.35)
\put(0.1,0){\includegraphics[width=9cm]{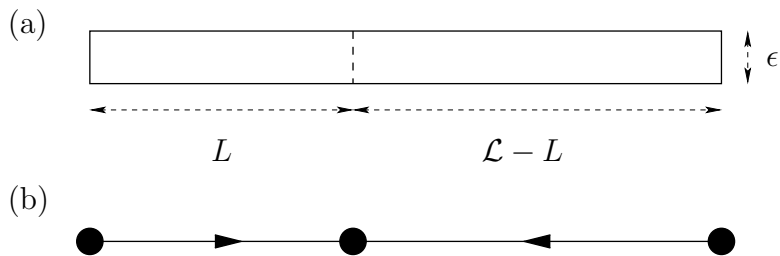}}
\put(1.12,0.28){$\epsilon$}
\put(0.3,0.14){$L$}
\put(0.7,0.14){$\cL-L$}
\put(0,0.33){(a)}
\put(0,0.07){(b)}
\end{picture}
\caption{(a) An $\epsilon$ thick piston.  (b) The corresponding piston graph.}\label{fig:piston}
\end{center}
\end{figure}

Specifically we consider the negative Laplace operator on the piston graph where $\psi_1(0)=\psi_2(0)=0$ and at the central vertex $\psi_1(L)=\psi_2(\cL-L)=\phi$ and
$\psi'_1(L)+\psi'_2(\cL-L)=\lambda \phi$.  Again $\phi$ is simply convenient notation to denote the value of the continuous graph function at the central vertex.  These boundary conditions define a self-adjoint operator on the graph.
Following Section \ref{sec:general star} we write the matching conditions at the center using the $\A, \B$ matrix pair, namely
\begin{equation}\label{eq:piston AB}
    \A=\left( \begin{array}{cc}
    1&-1 \\
    -\lambda & 0 \\
    \end{array} \right)  \qquad
    \B= \left( \begin{array}{cc}
    0&0 \\
    1& 1 \\
    \end{array} \right).
\end{equation}
Such matching conditions produce an energy dependent unitary scattering matrix for the central vertex,
\begin{eqnarray}
  S(k) &=& - (\A+\ui k\B)^{-1}(\A-\ui k\B) \\
   &=& \frac{1}{2\ui k -\lambda} \left( \begin{array}{cc}
    \lambda &2\ui k \\
    2 \ui k & \lambda \\
    \end{array} \right). \label{eq:piston S}
\end{eqnarray}

The transformed secular function $\hat{f}(t)$ was given in terms of the matching conditions conditions (\ref{eq:star gc f hatf}),
\begin{eqnarray}
    \hat{f}(t)&=  \det \left( \begin{array}{cc}
    \A & \B \\
    \UI_B & \frac{1}{t} \, \tanh (t\L) \\
    \end{array} \right)   \\
    &= -\frac{1}{t}\tanh Lt - \frac{1}{t} \tanh (\cL-L)t -\frac{\lambda}{t^2} \tanh Lt \tanh (\cL-L)t .\label{eq:piston fhat}
\end{eqnarray}
We see $a_1=-2$ and $a_2=-\lambda$ as $\hat{f}(t)\sim -2t^{-1} -\lambda t^{-2}$ in the large $t$ limit.
Substituting into the zeta function representation we get
\begin{eqnarray}
 \fl   \zeta(s)= \zeta_H\Big( 2s,\frac{1}{2} \Big) \pi^{-2s}\Big( L^{2s} +(\cL-L)^{2s}\Big) +
    \frac{\sin \pi s}{\pi} \left[ \int_0^1 t^{-2s} \frac{\ud}{\ud t} \log \hat{f}(t) \, \ud t \right. \nn \\
 \fl   \left. +\int_1^\infty t^{-2s} \frac{\ud}{\ud t} \left( \log \Big( \frac{t \hat{f}(t)} 2  \Big) -\frac{\lambda}{2\, t} \right) \, \ud t -\frac{1}{2s} -\frac{\lambda}{2 (2s+1)} \right] \textrm{ for } -1<\re \, s<1.
\end{eqnarray}
Evaluating the Casimir force we will differentiate with respect to $L$ keeping the total length of the graph $\cL$ fixed.  The Casimir force is then
\begin{equation}\label{eq:piston Ec}
    F_c= \frac{\pi \cL (2L-\cL)}{48 \, L^2(\cL-L)^2}
    +\frac{1}{2\pi} \int_0^\infty \frac{\partial}{\partial L} \log \hat{f}(t) \, \ud t \ .
\end{equation}
The Casimir force on the central vertex of the Piston is plotted for some representative  choices of the coupling parameter $\lambda$ in Figure \ref{fig:piston F}.  A positive force $F_c$ acts to increase $L$ so we see that the central vertex is attracted to the nearest node for all coupling constants $\lambda$.  Increasing $\lambda$ increases the magnitude of the attraction without  fundamentally changing its nature.  The effect rapidly saturates for values of $\lambda$ above $\lambda=100$.

\begin{figure}[htb]
\begin{center}
\setlength{\unitlength}{10cm}
\begin{picture}(1,0.6)
\put(0,0){\includegraphics[width=10cm]{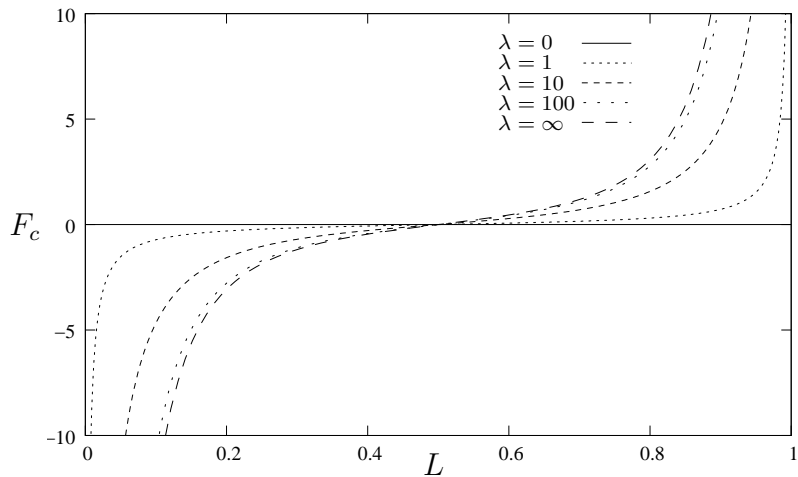}}
\put(-0.05,0.3){$F_c$}
\put(0.5,-0.02){$L$}
\put(0.6,0.547){\scriptsize $\lambda=0$}
\put(0.6,0.52){\scriptsize $\lambda=1$}
\put(0.6,0.493){\scriptsize $\lambda=10$}
\put(0.6,0.466){\scriptsize $\lambda=100$}
\put(0.6,0.439){\scriptsize $\lambda=\infty$}
\end{picture}
\caption{Casimir force on on the central vertex of the piston graph plotted for some representative choices of the coupling parameter $\lambda$. The total length of the piston $\cL=1$.  When $\lambda=0$ the derivatives on either side of the central vertex match and the vertex is invisible.  In the limit $\lambda \to \infty$ the graph decouples into two intervals with Dirichlet boundary conditions, $F_c=\pi\cL(2L-\cL)/24L^2(\cL-L)^2$.}
\label{fig:piston F}
\end{center}
\end{figure}

Note that the result for the piston graph does not contradict that for the Casimir force on a star graph with equal bond lengths.  As we have fixed the total length of the piston the force when the bond lengths are equal is clearly zero by symmetry.   Nevertheless the force on a star with two bonds, constrained to have equal lengths, is still attractive.

The asymptotic expansion of the heat kernel of the piston graph is determined by the $t\to \infty$ behavior of $\hat{f}(t)\sim -2/t -\lambda/t^2$.  Applying (\ref{detBez}) in this case the asymptotic expansion of the heat kernel as $t\to 0$ reads
\beqa
K(t) \sim \frac{{\cal L}} {\sqrt{4\pi t}}-\frac 1 2  - \sum_{k=1,3/2,2,...} ^\infty \frac{(-1)^{2k}}{(2k-1)\Gamma \left( k-\frac 1 2\right) } \left( \frac{\lambda}{2}\right)^{2k-1} t^{k-1/2} \ . \label{eq:piston K(t)}\eeqa

\section{Conclusions}
 In this article we have developed a systematic contour integral technique to analyze spectral zeta functions on graphs. The main strength of the formalism, apart from its simplicity, is that one can analyze general graphs with general vertex matching conditions in one calculation. Particular graphs and boundary conditions are extracted with ease.

The secular equation, together with the argument principle, are at the center of the approach. They allow to write down an integral representation for the zeta function valid in the half-plane $\Re s > 1/2$, see equations (\ref{eq:gen f}) and (\ref{eq:gen contour int}). The analytic continuation to the whole complex plane is obtained from an asymptotic behavior of the secular equation, see Section 8. As a result it is straightforward to evaluate the zeta determinant, equation (\ref{eq:gen zeta'0}), the Casimir force, Theorem \ref{thm:Fc graph}, and the heat kernel coefficients, equation (\ref{heatasygen}), for which we obtain new results in greater generality than those in the literature.
\ack
The authors would like to thank Gregory Berkolaiko, Jon Keating, Peter Kuchment, Jens Marklof, Robert Piziak and Brian Winn for helpful suggestions.
KK is supported by National Science Foundation grant PHY--0554849 and JMH is supported by National Science Foundation grant DMS--0604859.

\end{document}